\documentclass[envcountsame,envcountsect,orivec]{llncs}

\usepackage[latin1]{inputenc}
\usepackage{fixltx2e}
\usepackage{fix-cm}
\usepackage{xspace}
\usepackage{amsmath, amssymb}
\usepackage{listings}
\usepackage{graphicx}
\usepackage{url}
\usepackage{ifpdf}
\usepackage{microtype}
\usepackage[USenglish]{babel}

\ifpdf
\else
\fi

\spnewtheorem{property*}[theorem]{Property}{\bfseries\upshape}{\itshape}

\newcommand{\N}{\ensuremath{\mathbb{N}}\xspace}

\newcommand{\NTIME}[1]{\ensuremath{\mathrm{NTIME}(#1)}\xspace}

\newcommand{\DSPACE}[1]{\ensuremath{\mathrm{DSPACE}(#1)}\xspace}

\newcommand{\DLOGTIME}{\ensuremath{\mathrm{DLOGTIME}}\xspace}
\renewcommand{\P}{\ensuremath{\mathrm{P}}\xspace}

\newcommand{\PSPACE}{\ensuremath{\mathrm{PSPACE}}\xspace}
\newcommand{\PH}{\ensuremath{\mathrm{PH}}\xspace}

\newcommand{\AC}[1]{\ensuremath{\mathrm{AC^{#1}}}\xspace}
\newcommand{\ACC}[1]{\ensuremath{\mathrm{ACC^{#1}}}\xspace}
\newcommand{\TC}[1]{\ensuremath{\mathrm{TC^{#1}}}\xspace}
\newcommand{\NC}[1]{\ensuremath{\mathrm{NC^{#1}}}\xspace}
\newcommand{\sTC}[1]{\ensuremath{\mathrm{sTC^{#1}}}\xspace}

\newcommand{\REG}{\ensuremath{\mathrm{REG}}\xspace}
\newcommand{\CFL}{\ensuremath{\mathrm{CFL}}\xspace}
\newcommand{\DCFL}{\ensuremath{\mathrm{DCFL}}\xspace}
\newcommand{\CSL}{\ensuremath{\mathrm{CSL}}\xspace}
\newcommand{\VPL}{\ensuremath{\mathrm{VPL}}\xspace}

\newcommand{\APER}{\ensuremath{\mathrm{APER}}\xspace}
\newcommand{\NEUTRAL}{\ensuremath{\mathrm{NEUTRAL}}\xspace}

\newcommand{\FO}{\ensuremath{\mathrm{FO}}\xspace}

\newcommand{\MAJ}{\ensuremath{\mathrm{MAJ}}\xspace}

\newcommand{\arb}{\ensuremath{\mathbf{arb}}}
\newcommand{\reg}{\ensuremath{\mathbf{reg}}}
\newcommand{\bit}{\ensuremath{\textsc{Bit}}}
\newcommand{\even}{\ensuremath{\textsc{Even}}}

\newcommand{\bigO}[1]{\ensuremath{O(#1)}}
\newcommand{\co}{\ensuremath{\mathrm{co}}}
\newcommand{\withPolyAdvice}{\ensuremath{/\mathrm{poly}}}

\newcommand{\un}{\ensuremath{\mathrm{un}}}
\newcommand{\neutral}[1]{\ensuremath{\mathord{\mathrm{Neutral}(#1)}}}
\renewcommand{\complement}[1]{\ensuremath{\overline{#1}}}
\newcommand{\integer}[1]{\ensuremath{\widehat{#1}}}
\newcommand{\bc}[1]{\ensuremath{\mathord{\mathrm{BC}(\nobreak#1\nobreak)}}}
\newcommand{\powerset}[1]{\ensuremath{\mathord{\mathfrak{P}\!\left(\nobreak#1\nobreak\right)}}}
\newcommand{\VSTR}{\ensuremath{\mathrm{Struc}}\xspace}
\newcommand{\VSTRn}{\ensuremath{\mathrm{Struc}_n}\xspace}
\newcommand{\kernel}[1]{\ensuremath{\mathrm{kern}(#1)}}

\newcommand{\frakA}{\ensuremath{\mathfrak{A}}}
\newcommand{\calC}{\ensuremath{\mathcal{C}}}
\newcommand{\calL}{\ensuremath{\mathcal{L}}}
\newcommand{\calQ}{\ensuremath{\mathcal{Q}}}
\newcommand{\calU}{\ensuremath{\mathcal{U}}}
\newcommand{\calV}{\ensuremath{\mathcal{V}}}
\newcommand{\letterA}{\ensuremath{\mathrm{a}}}
\newcommand{\letterB}{\ensuremath{\mathrm{b}}}
\newcommand{\letterC}{\ensuremath{\mathrm{c}}}
\newcommand{\letterE}{\ensuremath{\mathrm{e}}}

\begin{document}

\abovedisplayskip 6pt plus 1pt minus 2pt 
\belowdisplayskip 6pt plus 1pt minus 2pt 

\title{
  Extensional Uniformity for Boolean Circuits\thanks{Supported
  in part by DFG VO 630/6-1, by the NSERC of Canada and by the
  (Qu\'ebec) FQRNT.}
}

\author{%
  Pierre McKenzie\inst{1} \and 
  Michael Thomas\inst{2} \and 
  Heribert Vollmer\inst{2}
}

\institute{%
  D\'{e}p.~d'informatique~et~de~recherche~op{\'e}rationnelle, Universit{\'e}~de~Montr{\'e}al, C.P.~6128, succ.~Centre-Ville, Montr{\'e}al~(Qu{\'e}bec), H3C~3J7~Canada \\ \protect\url|mckenzie@iro.umontreal.ca|%
  \and
Institut f\"ur Theoretische~Informatik, Leibniz Universit\"{a}t~Hannover, Appelstr.~4, 30167~Hannover, Germany \\ \protect\url|{thomas, vollmer}@thi.uni-hannover.de|%
}

\maketitle

\begin{abstract}
  Imposing an extensional uniformity condition on a non-uni\-form
  circuit complexity
  class $\calC$ means simply intersecting $\calC$ with a uniform 
  class $\calL$. By contrast, the usual intensional uniformity conditions
  require that a resource-bounded machine be
  able to exhibit the circuits in the circuit family
  defining $\calC$.
  We say that $(\calC,\calL)$ has the \emph{Uniformity Duality
  Property} if the extensionally uniform class $\calC\cap\calL$
  can be captured intensionally by means of adding so-called
  \emph{$\calL$-numerical
  predicates} to the first-order descriptive complexity apparatus 
  describing the connection language of the circuit family defining $\calC$.
  
  This paper exhibits positive instances and negative instances of the Uniformity Duality Property.
  
\end{abstract}

\section{Introduction} \label{sect:introduction}

  A family $\{ C_n\}_{n\geq 1}$ of Boolean circuits is
  \emph{uniform} if the way in which $C_{n+1}$ 
  can differ from $C_n$ is restricted.
  Generally, uniformity is imposed by requiring that some form of a
  resource-bounded constructor on input $n$ be able to fully or
  partially describe $C_n$ (see \cite{bor77,ruz81,all89,baimst90,la89} 
  or refer to \cite{vol99} for an overview).
  Circuit-based language classes can then be compared with
  classes that are based on a finite computing mechanism such as a Turing
  machine.
  
  Recall the gist of descriptive complexity. Consider the set of 
  words $w\in\{a,b\}^\star$ having no $b$ at an even position.
  This language is described by the $\FO[\mathord{<},\even]$ formula $\neg
  \exists i \big(\even(i) \wedge P_b(i)\big)$. In such a first-order formula,
  the variables range over positions in $w$, a predicate $P_\sigma$ for
  $\sigma\in\{a,b\}$ holds at $i$ iff $w_i=\sigma$, and
  a \emph{numerical} predicate, such as the obvious $1$-ary $\even$
  predicate here, holds at its arguments iff these arguments fulfill
  the specific relation.
  
  The following viewpoint has emerged \cite{baimst90,baim94,bela06} over two decades:
  \emph{when a 
    circuit-based language class is characterized using first-order descriptive
    complexity, the circuit uniformity conditions spring up in the logic
    in the form of restrictions on the set of numerical predicates allowed}.
  
  As a well studied example \cite{imm89,baimst90},
  $\FO[\mathord{<},+,\times] = \DLOGTIME$-uniform $\AC{0} \subsetneq \text{non-uniform } \AC{0} = \FO[\arb]$, 
  where the latter class is the class of languages definable by first-order formulae
  entitled to \emph{arb}itrary numerical predicates (we use a logic and the set of
    languages it captures interchangeably when this brings no
    confusion).
  
  In a related vein but with a different emphasis,
  Straubing \cite{str94} presents a beautiful account of the
  relationship between automata theory, formal logic and (non-uniform)
  circuit complexity. 
  Straubing concludes by expressing the proven fact that $\AC{0}
  \subsetneq  \ACC{0}$ and the celebrated conjectures that $\AC{0}[q]
  \subsetneq \ACC{0}$ and that $\ACC{0} \subsetneq \NC{1}$ as instances of
  the following conjecture concerning the class $\REG$ of regular
  languages:
  \begin{align} 
    {\cal Q}[\arb] \cap \REG &= {\cal Q}[\reg]. \label{conj:reg}
  \end{align}
  In Straubing's instances, $\cal Q$ is an appropriate set of quantifiers
  chosen from $\{\exists\}\cup\{\exists^{(q,r)} : 0 \leq r < q\}$
  and $\reg$ is the set of \emph{reg}ular numerical predicates, that is, 
  the set of those numerical predicates   of arbitrary arity definable in
  a formal sense by finite automata.
  We stress the point of view that intersecting
  $\{\exists\}[\arb] = \FO[\arb]$ with $\REG$ to form $\FO[\arb] \cap
  \REG$ in conjecture~\eqref{conj:reg} amounts to 
  imposing uniformity on the non-uniform class FO$[\arb]$.
  And once again, imposing uniformity has the effect of restricting the
  numerical predicates:   it is a proven fact that 
  $\FO[\arb] \cap \REG = \FO[\reg]$, and conjecture~\eqref{conj:reg} 
  expresses the hope that this phenomenon extends from $\{\exists\}$ 
  to other $\calQ$, which would determine much of the internal structure of 
  $\NC{1}$.
  We ask:
  \begin{enumerate}
    \item Does the duality between uniformity in a circuit-based class
    and numerical predicates in its logical characterization extend
    beyond $\NC{1}$?
    
    \item What would play the role of the regular numerical predicates in
    such a duality?
    
    \item Could such a duality help understanding classes such as
    the context-free languages in $\AC{0}$?
    
  \end{enumerate}

  To tackle the first question, we note that intersecting with
  $\REG$ is just one out of many possible ways in which one
  can ``impose uniformity''. Indeed, if
  $\calL$ is any uniform language class, one can replace $\calQ[\arb]
  \cap \REG$ by $\calQ[\arb] \cap \calL$ to get another uniform subclass
  of ${\cal Q}[\arb]$. For example, consider any ``formal language
  class'' (in the loose terminology used by Lange  when discussing
  language theory versus complexity theory \cite{la89}), such as the
  class $\CFL$ of context-free languages. Undoubtedly, $\CFL$ is a
  uniform class of languages. Therefore, the class $\calQ[\arb] \cap
  \CFL$ is another uniform class well worth comparing with $\calQ[\mathord{<},+]$
  or $\calQ[\mathord{<},+,\times]$. Of course, 
  $\FO[\arb] \cap \CFL$ is none other than the poorly understood class
  $\AC{0} \cap \CFL$, and when $Q$ is a quantifier given by some word problem of a nonsolvable group, $(\FO\mathord{+}\{Q\})[\arb] \cap \CFL$ is the poorly understood
  class $\NC{1}\cap\CFL$ alluded to 20 years ago \cite{ibjira88}.
  
  The present paper thus considers classes $\calQ[\arb] \cap \calL$ for
  various $\calQ$ and $\calL$. 
  To explain its title, we note that the constructor-based approach
  defines uniform classes by specifying their properties: such
  definitions are \emph{intensional} definitions.
  By contrast, viewing ${\cal Q}[\arb]\cap \REG$ as a uniform class
  amounts to an \emph{extensional} definition, namely one that selects the
  members of ${\cal Q}[\arb]$ that will collectively form
  the uniform class.
  In this paper we set up the extensional uniformity framework and we
  study classes $\calQ[\arb] \cap \calL$ for
  $\calQ \supseteq \{\exists\}$.

  Certainly, the uniform class $\calL$ will determine the class of
    numerical predicates we have to use when trying to capture
    $\calQ[\arb] \cap \calL$, as Straubing does for $\calL=\REG$, as
    an intensionally
    uniform class. A contribution of this paper is to provide a
    meaningful definition for the set $\calL^\N$ of 
    \emph{$\calL$-numerical predicates}. Informally, $\calL^\N$ is
    the set of relations over
    the natural numbers that are definable in the sense of Straubing \cite[Section III.2]{str94} 
    by a language over a singleton alphabet
    drawn from $\calL$. When $\calL$ is $\REG$, the $\calL$-numerical
    predicates are precisely Straubing's regular numerical predicates.

  Fix a set $\calQ$ 
  of monoidal or groupoidal quantifiers in the sense of \cite{baimst90,vol99,lamcscvo01}.
    (As prototypical examples, the reader unfamiliar with such
    quantifiers may think of the
    usual existential and universal quantifiers, of
    Straubing's ``there exist $r$ modulo $q$'' quantifiers, or of
    threshold quantifiers such as ``there exist a majority'' or ``there
    exist at least $t$''). We
    propose the \emph{Uniformity Duality
    Property for $(\calQ,\calL)$} as a natural generalization of
    conjecture~\eqref{conj:reg}:
  
    \medskip\noindent{\bfseries Uniformity Duality Property for ($\calQ,\calL$)}:
      $$\calQ[\arb]\cap \calL = \calQ[\mathord{<}, \calL^\N]\cap \calL.$$

  Barrington, Immerman and Straubing \cite{baimst90} have shown that
  $\calQ[\arb]$ equals $\AC{0}[\calQ]$, that is,
  non-uniform $\AC{0}$ with $\calQ$ gates.
  Behle and Lange \cite{bela06} have shown that $\calQ[\mathord{<},
  \calL^\N]$ equals $\FO[\mathord{<},\calL^\N]\text{-uniform}\ \AC{0}[\calQ]$,
  that is, uniform $\AC{0}[\calQ]$ where the direct
  connection language of the circuit families can be described by
  means of the logic $\FO[\mathord{<},\calL^\N]$. Hence the Uniformity Duality Property
  can be restated in circuit complexity-theoretic terms as follows: 

    \medskip\noindent{\bfseries Uniformity Duality Property for ($\calQ,\calL$), 2nd form}:
      $$\AC{0}[\calQ] \cap \calL \ \ =\ \ 
      \FO[\mathord{<},\calL^\N]\text{-uniform}\ \AC{0}[\calQ]
      \ \cap\ \calL.$$
  
      By definition, $\calQ[\arb]\cap \calL \supseteq \calQ[\mathord{<}, \calL^\N]\cap \calL$. 
    The critical question is whether the reverse inclusion holds.
    Intuitively, the Uniformity Duality Property states that the ``extensional
    uniformity induced by intersecting $\calQ[\arb]$ with $\calL$'' is a
    strong enough restriction imposed on $\calQ[\arb]$ to permit
    expressing the uniform class using the
    $\calL$-numerical predicates, or in other words: the
    extensional uniformity given by intersecting the
    non-uniform class with $\calL$ coincides with the
    intensional uniformity condition given by first-order
    logic with $\calL$-numerical predicates. 
    Further motivation for this definition of $\calQ[\mathord{<}, \calL^\N]\cap
    \calL$ is as follows:
    \begin{itemize}
      \item when constructors serve to define uniform classes, they have
      access to input lengths but not to the inputs themselves; a
      convenient logical analog to this is to use the unary alphabet
      languages from $\calL$ as a basis for defining the extra numerical
      predicates
      
      \item if the closure properties of $\calL$ differ from the 
      closure properties of $\calQ[\arb]$, then
      $\calQ[\arb]\cap \calL = \calQ[\mathord{<},
      \calL^\N]$ may fail trivially 
      (this occurs for example when $\calL=\CFL$ and 
      $\calQ=\{\exists\}$ since the non-context-free language 
      $\{\letterA^n \letterB^n \letterC^n: n\geq 0\}$
      is easily seen to belong to $\calQ[<,\calL^\N]$ by closure under 
      intersection of the latter); hence intersecting 
      $\calQ[\mathord{<}, \calL^\N]$ with $\calL$
      before comparing it with $\calQ[\arb]\cap \calL$
      is necessary to obtain a reasonable
      generalization of Straubing's conjecture for
      classes $\calL$ that are not Boolean-closed.
      
    \end{itemize}
  
  We now state our results, classified, loosely, as
  foundational observations~(F) or technical statements~(T).
  We let $\calL$ be \emph{any} class of languages.

  \begin{list}{*}{\itemsep \lineskip \settowidth{\leftmargin}{\indent} \settowidth{\labelwidth}{(T)} \settowidth{\labelsep}{\,}}

    \item[(F)] By design, the Uniformity Duality Property for $(\calQ,\REG)$
    is precisely Straubing's conjecture~\eqref{conj:reg}, hence its conjectured
    validity holds the key to the internal structure of $\NC{1}$. 
  
    \item[(F)] The Uniformity Duality Property for
    $(\{\exists\},\NEUTRAL)$
    is precisely the Crane Beach Conjecture \cite{baimlascth05}; here, $\NEUTRAL$ is the class of languages $L$ that have a neutral letter, i.e., a letter $e$ that may be arbitrarily inserted into or deleted from words without changing membership in $L$. The Crane
    Beach conjecture, stating that any
    neutral letter language in $\AC{0}=\FO[\arb]$ can be expressed in
    $\FO[\mathord{<}]$, was motivated by attempts to develop a purely
    automata-theoretic proof that \emph{Parity}, a neutral letter language, is
    not in $\AC{0}$. The Crane Beach Conjecture was ultimately refuted \cite{baimlascth05}, 
    but several of its variants have been studied. Thus \cite{baimlascth05}:
    \begin{itemize}
      \item the Uniformity Duality Property for $(\{\exists\}, \NEUTRAL)$ fails
      \item the Uniformity Duality Property for $(\{\exists\}, \NEUTRAL \cap \REG)$ holds
      \item the Uniformity Duality Property for $(\{\exists\}, \NEUTRAL \cap \{$two-letter lan\-gua\-ges$\})$ holds.
    \end{itemize}

    \item[(T)] 
    Our definition for the set $\calL^\N$ of
    $\calL$-numerical predicates parallels Straubing's definition of
    regular numerical predicates. 
    For kernel-closed language classes $\calL$ that are closed under homomorphisms, 
    inverse homomorphisms and intersection with a regular language, 
    we furthermore characterize $\calL^\N$ as the set of predicates expressible
    as one generalized unary $\calL$-quantifier applied to an FO$[\mathord{<}]$-formula.
    (Intuitively, $\calL$-numerical predicates are those
    predicates definable in first-order logic with one ``oracle call'' to a
    language from $\calL$.)   
    
    \item[(T)] We characterize the numerical predicates that surround the
    context-free languages:
      first-order combinations of $\CFL^\N$ suffice to capture all 
      semilinear predicates over $\N$; in particular, 
      $\FO[\mathord{<,+}] = \FO[\DCFL^\N] = \FO[\bc{\CFL}^\N]$,
      where $\DCFL$ denotes the deterministic context-free languages and 
      $\bc{\CFL}$ is the Boolean closure of $\CFL$.
    
    \item[(T)] We deduce that, despite the fact that $\FO[\bc{\CFL}^\N]$ contains all
    the semilinear relations, the Uniformity Duality Property fails for
    $(\{\exists\},\calL)$ in each of the following cases:
    \begin{itemize}
      \item $\calL =  \CFL$
      \item $\calL = \VPL$, the ``visibly pushdown languages'' recently introduced by \cite{alma04}
      \item $\calL =$ Boolean closure of the deterministic context-free languages
      \item $\calL =$ Boolean closure of the linear context-free languages
      \item $\calL =$ Boolean closure of the context-free languages.
    \end{itemize}
    The crux of the justifications of these negative results is a
    proof that the complement of the ``Immerman language'', used in
    disproving the Crane Beach Conjecture, is context-free.
  
    \item[(T)] At the opposite end of the spectrum, while it is clear that the Uniformity Duality
    Property holds for the set of all languages and any $\calQ$, we show that the Uniformity Duality
    Property already holds for $(\calQ,\calL)$ whenever $\calQ$ is a set of groupoidal quantifiers and $\calL= \NTIME{n}^\calL$; 
    thus it holds for, e.\,g., the rudimentary languages, $\DSPACE{n}$, $\CSL$ and $\PSPACE$.

  \end{list}
  
  The rest of this paper is organized as
  follows. Section~\ref{sect:preliminaries} contains preliminaries.
  Section~\ref{sect:duality_property} defines the $\calL$-numerical
  predicates and introduces the Uniformity Duality Property
  formally. The context-free numerical predicates are investigated in
  Section~\ref{sect:cfl_num_pred}, and the duality
  property for classes of context-free languages is considered in
  Section~\ref{sect:cfl_property}. Section~\ref{sect:csl_property}
  shows that the duality property holds when $\calL$ is ``large
  enough''. Section~\ref{sect:conclusion} concludes with a summary and a 
  discussion. 

\section{Preliminaries} \label{sect:preliminaries}

  \subsection{Complexity Theory}
  
  We assume familiarity with standard notions in formal languages,
  automata and complexity theory.
  
  When dealing with circuit complexity classes, all references will be
  made to the non-uniform versions unless otherwise stated.  Thus
  $\AC{0}$ refers of the Boolean functions
  computed by constant-depth polynomial-size unbounded-fan\-in
  $\{\vee,\wedge,\neg\}$-circuits.
  And $\DLOGTIME$-uniform $\AC{0}$ refers to the set of those
  functions in $\AC{0}$ computable by a circuit family having a 
  direct connection language decidable in time $\bigO{\log n}$ 
  on a deterministic Turing machine (cf. \cite{baimst90,vol99}). 
  
  \subsection{First-Order Logic}

  Let $\N$ be the natural numbers $\{1,2,3, \ldots\}$ and let $\N_0 = \N \cup \{0\}$.
  A \emph{signature} $\sigma$ is a finite set of relation symbols with fixed arity and constant symbols. 
  A $\sigma$-structure $\frakA = \langle \calU^\frakA, \sigma^\frakA \rangle$ consists of a set $\calU^\frakA$, called \emph{universe}, and a set $\sigma^\frakA$ that contains an \emph{interpretation} $R^\frakA \subseteq (\calU^\frakA)^k$ for each $k$-ary relation symbol $R \in \sigma$. We fix the interpretations of the ``standard'' numerical predicates $<$, $+$, $\times$, etc. to their natural interpretations. By $\bit$ we will denote the binary relation $\{ (x,i) \in \N^2 : \text{bit $i$ in the binary representation of $x$ is $1$}\}$.
  For logics over strings with alphabet $\Sigma$, we will use signatures extending $\sigma_\Sigma = \{ P_a : a \in \Sigma\}$ and identify $w = w_1 \cdots w_n \in \Sigma^\star$ with $\frakA_w = \langle \{1,\ldots, n\}, \sigma^{\frakA_w} \} \rangle$, where $P_a^{\frakA_w}= \{ i \in \N : w_i = a \}$ for all $a \in \Sigma$. We will not distinguish between a relation symbol and its interpretation, when the meaning is clear from the context.
  
  Let $\calQ$ be a set of (first-order) quantifiers.
  We denote by $\calQ[\sigma]$ the set of first-order formulae over $\sigma$ 
  using quantifiers from $\calQ$ only.  
  The set of all $\calQ[\sigma]$-formulae will be referred to as the 
  \emph{logic} $\calQ[\sigma]$. In case $\calQ = \{\exists\}$ 
  ($\calQ = \{\exists\}\cup\calQ'\}$), we will also write $\FO[\sigma]$ 
  ($\FO\mathord{+}\calQ'[\sigma]$, respectively).  When discussing 
  logics over strings, we will omit the relation symbols from $\sigma_\Sigma$.
  
  Say that a language $L \subseteq \Sigma^\star$ is \emph{definable}
  in a logic $\cal Q[\sigma]$ if there exists a
  $\calQ[\sigma]$-formula $\varphi$ such that $\frakA_w \models
  \varphi \Longleftrightarrow w \in L$ for all $w \in \Sigma^\star$,
  and say that a relation $R\subseteq\N^n$ is definable by a
  $\calQ[\sigma]$-formula if there exists a formula $\varphi$ with free
  variables $x_1,\ldots,x_n$ that defines $R$ for all sufficiently large
  initial segment of $\N$, i.\,e., if $\langle \{1,\ldots, m\}, 
  \sigma \rangle \models \varphi(c_1,\ldots,c_n) \iff (c_1,\ldots,c_n) \in R$ for all $m \geq c_\mathrm{max}$, 
  where $c_\mathrm{max}=\max\{c_1,\ldots,c_n\}$~\cite[Section~3.1]{schweikardt05}.
  By abuse of notation, we
  will write $L \in \calQ[\sigma]$ (or $R \in \calQ[\sigma]$) to express
  that a language $L$ (a relation $R$, resp.) is definable by a
  $\calQ[\sigma]$-formula and use a logic and the set of languages and
  relations it defines interchangeably.

\section{The Uniformity Duality Property} \label{sect:duality_property}

  In order to generalize conjecture \eqref{conj:reg}, we propose
  Definition~\ref{numerical} as a simple
  generalization of the regular numerical predicates defined
  using \calV-structures by Straubing \cite[Section III.2]{str94}.
  \begin{definition}
  Let $\calV_n=\{x_1,\ldots,x_n\}$ be a nonempty set of variables and
  let $\Sigma$ be a finite alphabet. A \emph{$\calV_n$-structure} is a
  sequence  
    $$
      \underline{w} = (a_1,V_1) \cdots (a_m,V_m) \in \left(\Sigma \times \powerset{\calV_n}\right)^\star
    $$
    such that $a_1, \ldots, a_m \in \Sigma$ and the nonempty sets
    among $V_1, \ldots, V_m$ form a partition of $\calV_n$
    (the underscore distinguishes $\calV_n$-structures from ordinary strings).
    Define $\Gamma_n= \{0\}\times \powerset{\calV_n}$. 
    We say that a
    $\calV_n$-structure $\underline{w}$ is \emph{unary} if 
    $\underline{w}\in \Gamma_n^\star$, i.\,e., if $a_1 \cdots a_n$ is defined 
    over the singleton alphabet $\{0\}$; in that case, we define the
    \emph{kernel} of $\underline{w}$, $\kernel{\underline{w}}$, as
    the maximal prefix of $\underline{w}$ that does not end with
    $(0,\emptyset)$; to signify that $x_i \in
    V_{c_i}$ for all $1\leq i \leq n$, we also write
    $\kernel{\underline{w}}$ as $[x_1=c_1,\ldots,x_n=c_n]$ and
    we let $\underline{w}^\N$ stand for $(c_1,\ldots, c_n)$.

    We define $\VSTRn$ as the language of all such words in $\Gamma_n^\star$ that are unary $\calV_n$-structures and let $\VSTR = \bigcup_{n>0} \VSTRn$.
  \end{definition}
  
  Any set $L$ of unary $\calV_n$-structures naturally prescribes a
  relation over the natural numbers. Hence, a set of such $L$
  prescribes a set of relations, or numerical predicates, over $\N$.

  \begin{definition} \label{numerical}
    Let $L\subseteq \Gamma_n^\star$ be a \emph{unary $\calV_n$-language}, that is, a set of
    unary $\calV_n$-structures. Let
    $L^\N = \{\underline{w}^\N : \underline{w} \in L \}$ denote
    the relation over $\N^n$ defined by $L$. 
    Then the \emph{$\calL$-numerical predicates} are defined as 
    $$
      \calL^\N = \{L^\N : L\in \calL \text{ and } L \subseteq \VSTR\}.
    $$
    We say that a language $L$ is \emph{kernel-closed} if, for every $\underline{w}\in L$, $\kernel{\underline{w}} \in L$. We further say that a language class $\calL$ is \emph{kernel-closed} if, for every $L \in \calL$ there exists an $L' \in \calL$ such that $L^\N=L'^\N$ and $L'$ is kernel-closed.
  \end{definition}
  
  \begin{remark}
    A unary $\calV_n$-language $L$ defines a unique numerical
    relation $L^\N$. Conversely, if two unary $\calV_n$-structures
    $\underline{v}$ and $\underline{w}$ define the same tuple
    $\underline{v}^\N=\underline{w}^\N$, then one of the two is
    obtained from the other by padding on the right with the letter
    $(0,\emptyset)$, i.e.,
    $\underline{v}\in\underline{w}(0,\emptyset)^\star$ or
    $\underline{w}\in\underline{v}(0,\emptyset)^\star$. Hence a numerical
    relation $R$ uniquely determines $\kernel{L}=
    \{\kernel{\underline{w}} : \underline{w}\in L\}$ for any language $L$ such that $R=L^\N$.
  \end{remark}

  We point out the following facts, where we write $\equiv_q r$ for the unary predicate $\{ x: x \equiv r \mod q \}$.
  
  \begin{proposition} \label{prop:regular-instances}
    Let $\APER$ and $\NEUTRAL$ denote the set of aperiodic languages and the set of languages having a neutral letter respectively. Then
    \begin{enumerate}
      \item $\APER^\N = \FO[\mathord{<}]$, \label{prop:regular-instances-aper}
      \item $\REG^\N = (\AC{0} \cap \REG)^\N = \FO[\mathord{<},\{ \equiv_q r : 0 \leq r < q \}] = \reg$, and \label{prop:regular-instances-reg}
      \item $\NEUTRAL^\N \subseteq \FO[\mathord{<}]$. \label{prop:regular-instances-neutral}
    \end{enumerate}
  \end{proposition}
 
  \begin{proof}
    For part~\ref{prop:regular-instances-aper}, let $L \in \APER$ with
    $L\subseteq \VSTRn$. Define $L' = \kernel{L}\cdot(0,\emptyset)^\star$. We claim that $L'$ is
    aperiodic. To see this, note that for any language $K$, any monoid
    recognizing $L$ also recognizes $LK^{-1}= \{v\in\Gamma_n^\star : vu\in L \mbox{ for some } u\in K\}$ \cite[Proposition 2.5]{pin84}.
    Hence $L[(0,\emptyset)^\star]^{-1}$ is
    aperiodic. But $\kernel{L} \subseteq L[(0,\emptyset)^\star]^{-1}$.
    So $L'$ equals $L[(0,\emptyset)^\star]^{-1}\cdot (0,\emptyset)^\star$ and is indeed aperiodic.
  
    Hence there exists a formula $\varphi \in \FO[<]$ such that $L(\varphi) = L'$ \cite{mcpa71}. Let
    $\psi(x_1,\ldots,x_n)$ be obtained from $\varphi$ by replacing each
    $P_{(0,V)}(x)$, $V \subseteq \calV_n$, with $\bigwedge_{x_i \in V} x=x_i \land \bigwedge_{x_i \notin V} x \neq x_i$.
    Then for all $\underline w = [x_1=c_1,\ldots,x_n=c_n](0,\emptyset)^j$ and all $m \geq \max\{c_i : 1 \leq i \leq n\}$,
    \[\frakA_{\underline{w}} \models \varphi \iff \langle \{1,\ldots,m\}, <^\frakA \rangle \models \psi(c_1,\ldots,c_n).\]
    Let $\vec x$ and $\vec c$ abbreviate $x_1,\ldots, x_n$ and $c_1,\ldots,c_n$, respectively, and let $c_\mathrm{max}=\max\{c_i : 1 \leq i \leq n\}$. Then 
    \begin{align*}
      \vec c \in L^\N 
      \implies & \exists i\colon [x_1=c_1,\ldots,x_n=c_n](0,\emptyset)^i \in L \\
      \implies &  [x_1=c_1,\ldots,x_n=c_n] \in \kernel{L} \\
      \implies & \forall i\colon [x_1=c_1,\ldots,x_n=c_n](0,\emptyset)^i \in L' \\
      \implies & \forall i\colon \frakA_{[x_1=c_1,\ldots,x_n=c_n](0,\emptyset)^i} \models \varphi \\
      \implies & \forall m \geq c_\mathrm{max} \colon \langle \{1,\ldots,m\} , <, = , \vec c \rangle \models \psi(\vec x)
    \intertext{ and }
      \vec c \notin L^\N 
      \implies & \forall i\colon [x_1=c_1,\ldots,x_n=c_n](0,\emptyset)^i \notin L \\
      \implies & [x_1=c_1,\ldots,x_n=c_n] \notin \kernel{L} \\
      \implies & \forall i\colon [x_1=c_1,\ldots,x_n=c_n](0,\emptyset)^i \notin L' \\
      \implies & \forall i\colon \frakA_{[x_1=c_1,\ldots,x_n=c_n](0,\emptyset)^i} \not\models \varphi \\
      \implies & \forall m \geq c_\mathrm{max}\colon \langle \{1,\ldots,m\} , <, = , \vec c \rangle \not\models \psi(\vec x)
    \end{align*}
    
    Hence $L^\N = \{ \vec c \in \N^n : (\forall m \geq c_\mathrm{max}) [ \langle \{1,\ldots,m\} , <, = , \vec c \rangle \models \psi(\vec x)] \}$. But $\psi(\vec x) \in \FO[\mathord{<}]$ and therefore $L^\N \in \FO[\mathord{<}]$.  
  
    For the other inclusion, let $R \in \N^n \cap \FO[<]$ via formula $\psi(\vec x)$. Define 
    $\varphi \equiv \exists x_1 \cdots \exists x_n \big(\psi(\vec x) \land \chi(\vec x)\big)$, where 
    \[
      \chi(\vec x) \equiv 
      \bigwedge_{1\leq i \leq n} \,
        \bigvee_{\substack{V \in \powerset{\calV_n}, \\ x_i \in V}} 
        \bigg(
          P_{(0,V)} (x_i)
          \land 
          \forall z 
          \Big(
            \bigvee_{\substack{V' \in \powerset{\calV_n}, \\ x_i \in V'}} 
            \big(
              P_{(0,V')} (z) \leftrightarrow z=x_i
            \big) 
          \Big)
        \bigg).
    \]
    The purpose of $\chi(\vec x)$ is to bind the variables in $\vec x$ to their respective values in a corresponding $\calV_n$-structure: 
    for a string $\underline{w}=(0,V_1)\cdots(0,V_m) \in \Gamma_n^\star$ and a tuple $\vec c$, 
    $\frakA_{\underline{w}} \models \chi(\vec c)$ holds iff 
      $m \geq \max\{c_i : 1 \leq i\leq n\}$, $\underline{w}$ is a unary $\calV_n$-structure, and
      for all $1\leq i \leq n$, $x_i \in V_{c_i}$.
    Alike the above equivalence, we obtain 
    \begin{align*}
      \vec c \in R 
      \iff & \forall m \geq c_\mathrm{max}\colon \langle \{1,\ldots,m\}, <, = , \vec c \rangle \models \psi(\vec x) \\ 
      \iff & \forall i \colon \frakA_{[x_1=c_1,\ldots, x_n=c_n](0,\emptyset)^i} \models \varphi.
    \end{align*}
    Let $L$ be the unary $\calV_n$-language $\{\underline{w} \in \VSTRn : \frakA_{\underline{w}} \models \varphi \}$. 
    Then $R=L^\N=\big(L(\varphi)\big)^\N$ with $\varphi \in \FO[<]$. Thus $M_L$ is finite and aperiodic and $L^\N \in \APER^\N$.
    
    Part~\ref{prop:regular-instances-reg} follows analogously from \cite[Theorems~III.1.1 and~III.2.1]{str94}.
  
    For Part~\ref{prop:regular-instances-neutral}, let $R \in
    \NEUTRAL^\N$. Then $R=L^\N$ for some neutral letter language
    $L\subseteq \VSTRn$. Assume that
    the neutral letter of $L$ is $(0,\emptyset)$, otherwise
    $L\subseteq \VSTRn$ implies that $L$ is empty. Since we can insert
    or delete $(0,\emptyset)$ in any word at will, $L$ is fully
    determined by the $(0,\emptyset)$-free words it contains, that is,
    \begin{displaymath}
    L = 
      \hskip-0.5em 
      \bigcup_{\substack{(V_1,V_2, \ldots,V_{k}) \text{ partitions } \calV_n \\ \text{ and } (0,V_1)(0,V_2) \cdots (0,V_{k})\in L}}
      \hskip-0.5em 
    (0,\emptyset)^\star (0,V_1)(0,\emptyset)^\star (0,V_2) \cdots (0,\emptyset)^\star(0,V_{k})(0,\emptyset)^\star.
    \end{displaymath}
    This is a finite union of regular aperiodic languages.
    Hence $L$ is regular aperiodic, and $L^\N$ is in $\FO[<]$ by
    Part~\ref{prop:regular-instances-aper}.
  \end{proof}
  
  Having discussed the $\calL$-numerical predicates, we can
  state the property expressing the dual facets of uniformity, namely,
  intersecting with an a priori uniform class on the one hand, and
  adding the corresponding numerical predicates to first-order logics
  on the other.
  
  \begin{property*}[Uniformity Duality {for ($\calQ$,$\calL$)}] \label{property:duality}
    Let $\calQ$ be a set of quantifiers and let $\calL$ be a language class, then
    $$
      \calQ[\arb] \cap \calL = \calQ[\mathord{<}, \calL^\N] \cap \calL. 
    $$ 
  \end{property*}

  As $\calQ[\arb]=\AC{0}[\calQ]$ \cite{baimst90} and $\calQ[\mathord{<},
  \calL^\N]=\FO[\mathord{<},\calL^\N]$-uniform $\AC{0}[\calQ]$ \cite{bela06}, 
  the above property equivalently states that
  $$\AC{0}[\calQ] \cap \calL =  \FO[\mathord{<},\calL^\N]\text{-uniform}\
  \AC{0}[\calQ] \cap\calL.$$ 
  
  As a consequence of
  Proposition~\ref{prop:regular-instances}\,(\ref{prop:regular-instances-aper}--\ref{prop:regular-instances-reg}),
  the Uniformity Duality Property is equivalent to the instances
  of the Straubing conjectures obtained by setting $\calQ$ and
  $\calL$ as we expect, for example $\calQ\subseteq \{\exists\} \cup \{
  \exists^{(q,r)}: 0 \leq r < q\}$ and $\calL= \REG$ yield
  exactly \eqref{conj:reg}.
  Similarly, as a consequence of
  Proposition~\ref{prop:regular-instances}\,(\ref{prop:regular-instances-neutral}),
  the Uniformity Duality Property is equivalent to the Crane Beach
  Conjecture if $\FO[\mathord{<}] \subseteq \calL$.
  Property~\ref{property:duality} is thus false when
  $\calQ=\{\exists\}$ and $\calL$ is the set $\NEUTRAL$ of all neutral letter
  languages. For some other classes, the Crane Beach Conjecture and thus
  Property~\ref{property:duality} hold: consider for example the case
  $\calL = \REG\cap\NEUTRAL$ \cite{baimlascth05}, or the case
  $\calQ=\{\exists\}$ and $\calL\subseteq \NEUTRAL\cap \FO[\mathord{+}]$.
  Accordingly the Uniformity Duality Property both generalizes the conjectures of Straubing et~al. and captures the intuition underlying the Crane Beach Conjecture. Encouraged by this unification, we will take a closer look at the Uniformity Duality in the case of first-order logic and context-free languages in the next section. 
  
  In the rest of this section, we present an alternative
  characterization of $\calL^\N$ using
  $\FO[\mathord{<}]$-transformations and unary Lindstr\"om
  quantifiers. This is further justification for our definition of $\calL$-numerical 
  predicates. The reader unfamiliar with this topic may
  skip to the end of Section~\ref{sect:duality_property}.
  
\subsection*{Digression: Numerical Predicates and Generalized Quantifiers}
  
  Generalized or Lindstr\"om quantifiers provide a very general yet
  coherent approach to extending the descriptive complexity of
  first-order logics \cite{lin66}. Since we only deal with unary
  Lindstr\"om quantifiers over strings, we will restrict our
  definition to this case. 
  \begin{definition}
    Let $\Delta = \{\letterA_1,\ldots,
    \letterA_t\}$ be an alphabet, $\varphi_1,\ldots, \varphi_{t-1}$ be
    $\FO[\mathord{<}]$-formulae, each with $k+1$ free variables
    $x_1,x_2,\ldots,x_k,y$, and let $\vec x$ abbreviate
    $x_1, x_2, \ldots, x_k$. 
    Further, let $\textsc{Struct}(\sigma)$ denote the set of finite structures 
    $\frakA = \langle \calU^\frakA, \sigma^\frakA \rangle$ over $\sigma$. 
    Then $\varphi_1,\ldots, \varphi_{t-1}$ define an \emph{$\FO[\mathord{<}]$-transformation}
    $$[\varphi_1(\vec x),\ldots, \varphi_{t-1}(\vec x)]\colon \textsc{struct}(\{<\nolinebreak,x_1,\ldots,x_k\}) \to
    \Delta^\star$$ as follows: 
    Let $\frakA \in \textsc{struct}(\{<\nolinebreak,x_1,\ldots,x_k\})$, $x_i^\frakA =c_i \in \calU^\frakA$, $1\leq i \leq k$, and
    $s=|\calU^\frakA|$, then $[\varphi_1(\vec x),\ldots,
    \varphi_{t-1}(\vec x)](\frakA) = v_1\cdots v_s \in \Delta^\star$,
    where 
    $$
      v_i = 
      \begin{cases}
        \letterA_1, & \text{ if } \frakA \models \varphi_1(c_1,\ldots ,c_k, i), \\
        \letterA_j, & \text{ if } \frakA \models \varphi_j(c_1,\ldots ,c_k, i) \land \bigwedge_{l=1}^{j-1} \lnot\varphi_l(c_1,\ldots ,c_k, i), 1 < j < t,\\
        \letterA_t, & \text{ if } \frakA \models \bigwedge_{l=1}^{t-1} \lnot\varphi_l(c_1,\ldots ,c_k, i).
      \end{cases}
    $$      
    A language $L \subseteq \Delta^\star$ and an $\FO[\mathord{<}]$-transformation $[\varphi_1(\vec x),\ldots, \varphi_{t-1}(\vec x)]$ now naturally define a \emph{(unary) Lindstr\"om quantifier} $\calQ_L^\un$ via
    \[\frakA \models \calQ_L^\un y [\varphi_1(\vec x,y),\ldots, \varphi_{t-1}(\vec x,y)] \iff [\varphi_1(\vec x),\ldots, \varphi_{t-1}(\vec x)](\frakA) \in L.\]
    Finally, the set of relations definable by formulae $\calQ^\un_L y [\varphi_1(\vec x,y),\ldots, \varphi_{t-1}(\vec x,y)]$, where $L \in \calL$ and $\varphi_1,\ldots,\varphi_{t-1} \in \FO[\mathord{<}]$, will be denoted by $\calQ^\un_\calL\FO[\mathord{<}]$.
  \end{definition}
  
  The notation $[\varphi_1(\vec x),\ldots, \varphi_{t-1}(\vec x)]$ is chosen to distinguish the variables in $\vec x$ from $y$; the variables in $\vec x$ are interpreted by $\frakA$ whereas $y$ is utilized in the transformation.

  \begin{theorem} \label{thm:num-pred-by-lindstrom-q}
    Let $\cal L$ be a kernel-closed language class which is closed under
    homomorphisms, inverse homomorphisms and intersection with regular
    languages, then $\calL^\N = \calQ^\mathrm{un}_{\calL}\FO[\mathord{<}]$; that
    is, the $\calL$-numerical predicates correspond to the predicates
    definable using a unary Lindstr\"om quantifier over $\calL$ and an
    $\FO[\mathord{<}]$-transformation.
  \end{theorem}
  
  \begin{proof}
    For the inclusion from left to right, consider a relation $L^\N
    \subseteq \N^n$ in $\calL^\N$ and let $\calV_n=\{x_1,\ldots,
    x_n\}$. Define the $\FO[\mathord{<}]$-transformation $[\varphi_1(\vec
    x),\ldots,\varphi_{2^n-1}(\vec x)]$ from $\N^n$ to unary $\cal
    V$-structures as follows. For $1 \leq i < 2^n$, let  
    $$
      \varphi_i \equiv \bigwedge_{j \in V_i} (y = x_j) \land \bigwedge_{j \notin V_i} (y \neq x_j),
    $$ 
    where $V_i$ denotes the $i$th subset of $\cal V$ in the natural subset ordering, and associate the letter $(0,V_i)$ with $\varphi_i$.
    Let $\vec c= (c_1,\ldots c_n) \in \N^n$ and denote by $\frakA$ the
    structure $(\{1,\ldots,l\},\sigma)$ with $l=\max_i\{c_i\}$
    and $\{<,x_1,\ldots,x_n\}\subseteq \sigma$.
    Then $[\varphi_1(\vec
    x),\ldots,\varphi_{2^n-1}(\vec x)]$ maps $\frakA$ to the unary
    $\calV_n$-structure $v_1 \cdots v_l$. Then
    \begin{eqnarray*}
      \vec c \in L^\N & \Longrightarrow & \exists j\exists 
      \underline{w}=[x_1=c_1,\ldots, x_n=c_n]\colon \underline{w}(0,\emptyset)^j\in L \\
      & \Longrightarrow & [x_1=c_1,\ldots, x_n=c_n] \in L \ \ \mbox{
      (since $L$ is kernel-closed)}\\   
      & \Longrightarrow & [\varphi_1(\vec
      x),\ldots,\varphi_{2^n-1}(\vec x)](\frakA) \in L \\
      & \Longrightarrow & \frakA \models \calQ^\mathrm{un}_L y
      [\varphi_1(\vec x,y),\ldots,\varphi_{2^n-1}(\vec x,y)],
    \end{eqnarray*}
    and the reverse implications hold for any unary $\calV_n$-language $L$.

    For the opposite inclusion, let $R \subseteq \N^n \cap
    \calQ^\mathrm{un}_L\FO[\mathord{<}]$ via the transformation $[\varphi_1(\vec
    x),\ldots,\varphi_{k-1}(\vec x)]$ and the language $L \subseteq
    \Delta^\star \cap \calL$, $|\Delta|=k$.
    We must exhibit a unary $\calV_n$-language $A\in\calL$ such that
    $R=A^\N$, i.\,e., $R=h(\kernel{A})$, where
    \begin{align*}
      h\colon \{ \kernel{\underline{w}} : \underline{w} \mbox{ is a unary $\calV_n$-structure}\} &\rightarrow \N^n\\
      [x_1=c_1,\ldots, x_n=c_n]&\mapsto (c_1,\ldots,c_n).
    \end{align*}
    Our proof is similar to the proof of Nivat's Theorem, see \cite[Theorem~2.4]{lamcscvo01}.

    Define $B \subseteq
    \big( \Gamma_n \times \Delta \big)^\star$
    to consist of all words
    $\binom{\underline{u}_1}{v_1}\cdots\binom{\underline{u}_l}{v_l}$ such that
    $\underline{u}=\underline{u}_1 \cdots \underline{u}_l$ is the kernel of a unary
    $\calV_n$-structure and 
    $[\varphi_1(\vec x),\ldots,\varphi_{k-1}(\vec x)]$ maps
    $\underline{u}^\N$ to $v=v_1 \cdots v_l$. 
    Define the length-preserving homomorphisms 
    $f\colon \binom{\underline{u}}{v} \mapsto \underline{u}$ 
    and 
    $g\colon \binom{\underline{u}}{v} \mapsto v$.
    Then $R = (h\circ \mathrm{kern}\circ f)\big(B \, \cap\,
  g^{-1}(L)\big)$.

    We claim that $B$ is
    regular. Theorem~\ref{thm:num-pred-by-lindstrom-q} then follows
    from the closure properties of $\calL$ by setting $A=f\big(B \,
    \cap\, g^{-1}(L)\big)$
    
    For $1\leq i \leq k-1$, let $\varphi_i'(y)$ be defined as $\varphi_i'(y) \equiv \exists x_1 \cdots \exists x_n ( \varphi_i(\vec x,y) \land \psi(\vec x) \land \pi(\vec x))$, where 
    $\psi(\vec x)$ and $\pi(\vec x)$ bind the variables in $\vec x$ to their respective values in a corresponding unary $\calV_n$-structure
    and asserts that each variable occurs exactly once; that is,
    \begin{align*}
      \psi(\vec x) \equiv &\,
        \bigwedge_{1\leq i \leq n} \,
        \bigvee_{d \in \Delta} \,
        \bigvee_{\substack{V \in \powerset{\calV_n}, \\ x_i \in V}}
        P_{\binom{(0,V)}{d}} (x_i),
        \\
      \pi(\vec x) \equiv &\,
        \bigwedge_{1\leq i \leq n}
        \forall z 
          \Big(
            \big(
              \bigvee_{d \in \Delta} \,
              \bigvee_{\substack{V \in \powerset{\calV_n}, \\ x_i \in V}}
              P_{\binom{(0,V)}{d}} (z) 
            \big) \leftrightarrow z=x_i
          \Big).
    \end{align*}
    Now let  
    $\chi_j(z) \equiv \bigvee_{V \in \powerset{\calV_n}}
    P_{\binom{(0,V)}{d_j}}(z), 1 \leq j \leq k$, where $d_j$ is the
    $j$th letter in $\Delta$. Then a string
    $\binom{\underline{u}_1}{v_1}\cdots\binom{\underline{u}_l}{v_l} \in \big( \Gamma_n
    \times \Delta \big)^\star$ is in $B$ if and only if
    $\underline{u}=\underline{u}_1\cdots \underline{u}_l$ is the kernel of a unary
    $\calV_n$-structure and
    $$
      \forall z 
      \bigg(
        \bigwedge_{i=1}^{k-1}
        \Big( \big( \varphi_i'(z) \land \bigwedge_{l=1}^{i-1} \neg\varphi_l'(z) \big) \leftrightarrow \chi_i(z) \Big) 
        \land 
        \Big( \bigwedge_{l=1}^{k-1} \neg\varphi_l'(z) \leftrightarrow \chi_k(z) \Big) 
      \bigg)
    $$
    holds on $\vec z=\underline{u}^\N$, where the empty conjunction is defined to be true. Concluding, $B \in \FO[\mathord{<}] \subset \REG$.
  \end{proof}

  We stress that the above result provides a logical characterization
  of the $\calL$-numerical predicates for all kernel-closed classes $\calL$ forming
  a cone, viz. a class of languages $\calL$ closed
  under homomorphisms, inverse homomorphisms and intersection with
  regular languages \cite{gigrho69}. As the closure under these operations is
  equivalent to the closure under rational transductions
  (i.\,e., transductions performed by finite automata \cite{ber79}), we
  obtain: 
  
  \begin{corollary}
    Let $\calL$ be kernel-closed and closed under rational transductions, then 
    $\calL^\N = \calQ^\mathrm{un}_{\calL}\FO[\mathord{<}]$.
  \end{corollary}

\section{Characterizing the Context-Free Numerical Predicates} \label{sect:cfl_num_pred}

  In order to examine whether the Uniformity Duality Property for first-order logics holds in the case of context-free languages, we first need to consider the counterpart of the regular numerical predicates, that is, $\CFL^\N$. Our results in this section will relate $\CFL^\N$ to addition w.\,r.\,t. to first-order combinations, and are based upon a result by Ginsburg \cite{ginsburg66}. Ginsburg showed that the number of repetitions per fragment in bounded context-free languages corresponds to a subset of the semilinear sets. 
  For a start, note that addition is definable in $\DCFL^\N$. 
  
  \begin{lemma} \label{lem:add-in-cfl^n}
    Addition is definable in $\DCFL^\N$.
  \end{lemma}

  \begin{proof}
    Let $\calV_3=\{x_1,x_2,x_3\}$ and $L_+$ be a unary $\calV_n$-language defining addition, that is, $L_+^\N= \{ (x_1,x_2,x_3) : x_1+x_2 = x_3 \}$. Then $L_+$ is recognized by the following deterministic PDA $P=(Q,\Gamma_n, \Delta, \delta, q_0,\bot, \{q_\text{acc}\})$, where
    $Q=\{q_0, q_\text{acc} \}$, $\Delta=\{\bot,0\}$ and $\delta$ is defined as follows:
    \[
    \begin{array}{@{\delta(}l@{,\, }l@{,\, }r@{)\ = (}l@{,\, }r@{) \qquad}l@{ \delta(}l@{,\, }l@{,\, }r@{)\ = (}l@{,\, }r@{) }}
      z_0    &(0,\emptyset) &\gamma  & z_0          & 0\gamma && z_y    &(0,\emptyset) &\gamma  & z_y          & \gamma \\
      z_0    &(0,\{x\})     &\gamma  & z_x          & 0\gamma && z_y    &(0,\{x\})     &\gamma  & z_{xy}       & \gamma \\
      z_0    &(0,\{y\})     &\gamma  & z_y          & 0\gamma && z_{xy} &(0,\emptyset) &0       & z_{xy}       & \varepsilon \\
      z_x    &(0,\emptyset) &\gamma  & z_x          & 0\gamma && z_{xy} &(0,\{z\})     &0       & z_{z}        & \varepsilon \\
      z_x    &(0,\{y\})     &\gamma  & z_{xy}       & \gamma  && z_{z}  &\hphantom{(}\varepsilon   &\bot    & z_\text{acc} & \bot   \\
    \end{array}
    \] where $\gamma \in \Delta$.
  \end{proof}

  Next, we restate the result of Ginsburg in order to prepare ground for the examination of the context-free numerical predicates.
  In the following, let $w^\star$ abbreviate $\{w\}^\star$ and say that a language $L \subseteq \Sigma^\star$ is \emph{bounded} if there exists an $n \in \N$ and $w_1,\ldots, w_n \in \Sigma^+$ such that $L \subseteq w_1^\star \cdots w_n^\star$. 
    
  \begin{definition}
    A set $R \subseteq \N_0^n$ is \emph{stratified} if 
    \begin{enumerate}
      \item each element in $R$ has at most two non-zero coordinates, 
      \item there are no integers $i,j,k,l$ and $x=(x_1,\ldots, x_n), x'=(x_1',\ldots, x_n')$ in $R$ such that 
      $1\leq i < j < k < l \leq n$ and $x_ix_j'x_kx_l' \neq 0$.
    \end{enumerate}
    Moreover, a set $S \subseteq \N^n$ is said to be \emph{stratified semilinear} 
    if it is expressible as a finite union of linear sets, each with a stratified set of periods;
    that is, $S = \bigcup_{i=1}^m \{ \vec \alpha_{i0} + \sum_{j=1}^{n_i} k\cdot\vec{\alpha}_{ij} : k\in \N_0 \}$, 
    where $\vec\alpha_{i0} \in \N^n$, $\vec\alpha_{ij} \in \N_0^n$, $1 \leq j \leq n_i$, $1\leq i \leq m$, and each $P_i=\{ \vec\alpha_{ij} : 1 \leq j \leq n_i \}$ is stratified.
    
  \end{definition}
  
  \begin{theorem}[{\cite[Theorem 5.4.2]{ginsburg66}}] \label{thm:ginsburg}
    Let $\Sigma$ be an alphabet and $L \subseteq w_1^\star \cdots w_n^\star$ be bounded by $w_1, \ldots, w_n \in \Sigma^+$. Then $L$ is context-free if and only if the set 
    \[
      E(L) = \big\{(e_1,\ldots, e_n) \in  \N_0^n : w_1^{e_1}\ldots w_n^{e_n} \in L \big\}
    \]
    is a stratified semilinear set.
  \end{theorem}
  
  Theorem~\ref{thm:ginsburg} relates the bounded context-free languages to a strict subset of the semilinear sets. 
  The semilinear sets are exactly those sets definable by $\FO[\mathord{+}]$-formulae.   
  There are however sets in $\FO[+]$ that are undefinable in $\CFL^\N$: e.\,g., if $R=\{(x,2x,3x) : x \in \N\}$ 
  was definable in $\CFL^\N$ then $\{\letterA^n\letterB^n\letterC^n : n \in \N\} \in \CFL$.
  Hence, $\FO[\mathord{+}]$ can not be captured by $\CFL^\N$ alone. Yet, addition is definable in $\CFL^\N$, 
  therefore we will in the following investigate the relationship between
  first-order logic with addition, $\FO[\mathord{+}]$, and the Boolean closure of
  $\CFL$, $\bc{\CFL}$. 
  
  \begin{lemma} \label{lem:a_bounded}
    Let $R \subseteq \N^n$ and let $R=L^\N$ for language $L$. Then $L$ is bounded.
  \end{lemma}
  
  \begin{proof}
    Let $R \subseteq \N^n$ and $L$ such that $R=L^\N$. Let $L \subseteq \Gamma_n^\star$, where $\calV_n=\{x_1,\ldots,x_n\}$ is a set of variables. Every unary $\calV_n$-structure $\underline{w} \in L$ of a given order type $t \in T$ belongs to the language
    $$
    L_t = \left(0,\emptyset \right)^\star
          \left(0,V_1 \right)
          \left(0,\emptyset \right)^\star
          \left(0,V_2 \right)
          \cdots
          \left(0,V_k \right)
          \left(0,\emptyset \right)^\star,
    $$
    where $(V_1, V_2, \ldots, V_k)$ is an ordered partition of $\calV_n$. The number $p$ of such ordered partitions depends on $n$; in particular, $p$ is finite. Hence $L \subsetneq L_1^\star L_2^\star \cdots L_{t_p}^\star$, where $t_1,t_2,\ldots,t_p$ exhaust the possible types. Since each letter $(0,V_i)$ belongs to the language $\prod_{V \subseteq \calV_n} (0,V)^\star$, it follows that $L \subseteq \prod_{1 \leq i \leq p} \left( \prod_{V \subseteq \calV_n} (0,V)^\star \right)^{2n+1}$.
  \end{proof}
  
  \begin{lemma} \label{lem:cfl+bounded->fo[+]}
    All bounded context-free languages are definable in $\FO[\mathord{+}]$.
  \end{lemma}
  
  \begin{proof}
    Since $L$ is bounded, there exist $w_1,\ldots, w_n \in \Sigma^+$ such that $L \subseteq w_1^\star\ldots w_n^\star$. By Theorem~\ref{thm:ginsburg}, it holds that the set 
    $$
      E(L)=\{(e_1,\ldots, e_n) : w_1^{e_1}\ldots w_n^{e_n} \in L\}
    $$ 
    is stratified semilinear. It follows by semilinearity alone that, for $\vec{e}=(e_1,\ldots,e_n)$, $E(L) = \bigcup_{i=1}^m \{ \vec{e} : \vec{e} = \vec \alpha_{i0} + \sum_{j=1}^{n_i} k\cdot\vec{\alpha}_{ij}, k\in \N_0 \}$, where $\vec\alpha_{i0} \in \N^n$, $\vec\alpha_{ij} \in \N_0^n$ for all $1 \leq j \leq n_i$, $1 \leq i \leq m$. As a consequence, $L$ is defined via the formula 
    $\varphi = \exists e_1 \cdots \exists e_n \big( \varphi_{w_1,\ldots, w_n}(e_1, \ldots, e_n) \land \varphi_\text{lin}(e_1,\ldots,e_n) \big)$,
    where $\varphi_{w_1,\ldots, w_n}(e_1, \ldots, e_n)$ checks whether the input is of the form $w_1^{e_1}\cdots w_n^{e_n}$ and
    \[
      \varphi_\text{lin}(e_1,\ldots,e_n)\equiv \bigvee_{i=1}^m \Big(\exists a_{i,1} \cdots \exists a_{i,n_i}  \bigwedge_{k=1}^n\big( e_k = \alpha_{i0k} + \sum_{j=1}^{n_i} a_{ij}\alpha_{ijk} \big) \Big).
    \]
    In particular, $\varphi_{w_1^{e_1}\ldots w_n^{e_n}} \in \FO[\mathord{+}]$. Concluding, $L \in \FO[\mathord{+}]$.
  \end{proof}
  
  \begin{lemma}\label{lem:a->a^n}
    Let $\calV_n=\{x_1,\ldots,x_n\}$ and let $L$ be a unary $\calV_n$-language. Then $L \in \FO[\mathord{+}]$ implies $L^\N \in \FO[\mathord{+}]$.
  \end{lemma}
  
  \begin{proof}
    Since $L \in \FO[\mathord{+}]$, there exists a formula $\varphi \in \FO[\mathord{+}]$ such that for all $\underline{w}=(0,V_1) \cdots $ $(0,V_m) \in \Gamma_n^\star$,
    $$ \langle \{1,\ldots,m\}, <\nolinebreak, +, \mathbf{P} \rangle \models \varphi \iff \underline{w} \in L, $$
    where $\mathbf{P}=\{\mathrm{P}_{(0,V)} : V \in \powerset{\calV_n} \}$ and $\mathrm{P}_{(0,V)} (z)$ is true if and only if $z \in V$, for $1\leq z \leq m$.
    Let $\vec y = (y_1,\ldots, y_n)$. We construct the formula $\varphi'(\vec y)$ from $\varphi$ by replacing, for each variable $z$ and each $V \in \powerset{\calV_n}$, the predicate $P_{(0,V)}(z)$ with $\bigwedge_{y_i \in V} z=y_i \land \bigwedge_{y_i \notin V} z\neq y_i$.
    Then
    $$ \langle \{1, \ldots, m\}, <\nolinebreak, + \rangle \models \varphi'(\vec y) \iff \langle \{1,\ldots,m\}, <\nolinebreak, +, \mathbf{P} \rangle \models \varphi. $$
    Hence the formula $\varphi'$ witnesses $L^\N \in \FO[\mathord{+}]$.
  \end{proof}
  
  \begin{theorem} \label{thm:cfl subset fo+}
    $\bc{\CFL^\N} \subseteq \bc{\CFL}^\N \subseteq \FO[\mathord{+}].$
  \end{theorem}

  \begin{proof}
    The inclusion $\bc{\CFL^\N} \subseteq \bc{\CFL}^\N$ follows from 
      (1.) the fact that for every unary $\calV_n$-language $L \in \CFL$, there exists an equivalent kernel-closed unary $\calV_n$-language $L' \in \CFL$;
      (2.) $L_1^\N \cap L_2^\N = (L_1 \cap L_2)^\N$ and $L_1^\N \cup L_2^\N = (L_1 \cup L_2)^\N$ for kernel-closed languages $L_1, L_2$; and 
      (3.) the observation that for a kernel-closed $L\in \CFL$ with $R=L^\N$, the language $L'=L(0,\emptyset)^\star \in \CFL$ also verifies $R=L'^\N$, so that $\overline{R} = (\overline{L'} \cap \VSTRn)^\N \in \bc{\CFL}^\N$.
    
    It remains to show that $\bc{\CFL}^\N \subseteq \FO[\mathord{+}]$.
    Denote by $\co\CFL$ set of languages whose complement is in $\CFL$, i.\,e., $\co\CFL = \{ L : \overline{L} \in \CFL\}$.
    Let $R \subseteq \bc{\CFL}^\N\cap\N^n$ and $\calV_n=\{x_1,\ldots, x_n\}$. 
    Further, let $L$ be some unary $\calV_n$-language such that $R=L^\N$; w.\,l.\,o.\,g. $L = \bigcup_{i=1}^m \bigcap_{j=1}^{k_i} L_{ij}$ where $L_{ij} \in \CFL \cup \co\CFL$. 
    It holds that $L$ is a unary $\calV_n$-language if and only if $\bigcap_{j=1}^{k_i} L_{ij}$ is a unary $\calV_n$-language for all $1 \leq i \leq m$. Hence, we need to show that $L_i^\N$ is $\FO[\mathord{+}]$-definable, for any unary $\calV_n$-language $L_i=\bigcap_{j=1}^{k_i} L_{ij}$ with $L_{ij} \in \CFL \cup \co\CFL$.
    
    Note that $\VSTRn$ is bounded and definable in $\FO[<] \subset \REG$. Then
    $L_i = (L_{i1} \cap \VSTRn) \cap \cdots \cap (L_{ik} \cap \VSTRn)$ and each $(L_{ij} \cap \VSTRn)$, $1\leq j \leq k$, is bounded (Lemma~\ref{lem:a_bounded}). We have to distinguish the following two cases:
    
    \renewcommand{\descriptionlabel}[1]{\hspace\labelsep\normalfont{#1}}
    \begin{description}
      \item[Case 1: $L_{ij} \in \CFL$.] Then $L_{ij} \cap \VSTRn \in \CFL$. Thus Lemma~\ref{lem:cfl+bounded->fo[+]} implies $L_j \cap \VSTRn \in \FO[\mathord{+}]$. 
      
      \item[Case 2: $L_{ij} \in \co\CFL$.] As $L_{ij} \cap \VSTRn$ is bounded, it can be written as $L_{ij} \cap \VSTRn = \{w_1^{e_1} \cdots w_n^{e_n} : (e_1,\ldots, e_n) \in X\}$ for words $w_1, \ldots, w_n$ and some relation $X \subseteq \N_0^n$.
      Hence, 
      \[
         \overline{L_{ij} \cap \VSTRn} 
         = \overline{L_{ij}} \cup \overline{\VSTRn} 
         = ( \overline{L_{ij}} \cap \VSTRn ) \cup \overline{\VSTRn} 
       \]
      where $\overline{L_{ij}} \cap \VSTRn$ is the intersection of a context-free language with a regular language and therefore context-free. 
      Further note that $\overline{L_{ij}} \cap \VSTRn = \{w_1^{e_1} \cdots w_n^{e_n} \in \VSTRn : (e_1,\ldots, e_n) \notin X\}$ is bounded.
      From Lemma~\ref{lem:cfl+bounded->fo[+]} it now follows that $\overline{L_{ij}} \cap \VSTRn \in \FO[\mathord{+}]$. Thus, finally, $L_{ij} \cap \VSTRn = \overline{(\overline{L_{ij}} \cap \VSTRn) \cup \overline{\VSTRn}} \in \FO[\mathord{+}]$.
    \end{description}
    
    Summarizing, $L_i = (L_{i1} \cap \VSTRn) \cap \cdots \cap (L_{ik} \cap \VSTRn)$ is definable in $\FO[\mathord{+}]$ using the conjunction of the defining formulae. Since $L_i$ is a unary $\calV_n$-language by assumption, Lemma~\ref{lem:a->a^n} implies the claim. 
  \end{proof}
  
  That is, the relations definable in the Boolean closure of the context-free unary
  $V_n$-languages are captured by $\FO[\mathord{+}]$. Hence, $\FO[\bc{\CFL}^\N] \subseteq \FO[\mathord{+}]$. Now Lemma~\ref{lem:add-in-cfl^n} yields the following corollary.
  
  \begin{corollary} \label{cor:bccfl subset fo+}
    $\FO[\DCFL^\N] = \FO[\CFL^\N] = \FO[\bc{\CFL}^\N] = \FO[\mathord{+}].$
  \end{corollary}
  
  We note that in particular, for any $k \in \N$, the inclusion $\left(\bigcap_k\CFL\right)^\N \subsetneq \FO[\mathord{+}]$ holds, where $\bigcap_k\CFL$ denotes the languages definable as the intersection of $\leq k$ context-free languages: this is deduced from embedding numerical predicates derived from the infinite hierarchy of context-free languages by Liu and Weiner into $\CFL^\N$ \cite{liwe73}. Hence, 
  \[ 
    \textstyle \CFL^\N \subsetneq \cdots \subsetneq (\bigcap_k\CFL)^\N  \subsetneq (\bigcap_{k+1}\CFL)^\N \subsetneq \cdots \subsetneq (\bigcap\CFL)^\N \subseteq \FO[\mathord{+}].
  \]
  
  Unfortunately, we could neither prove nor refute $\FO[\mathord{+}] \subseteq \bc{\CFL}^\N$. The difficulty in comparing $\FO[\mathord{+}]$ and $\bc{\CFL}^\N$ comes to some extent from the restriction on the syntactic representation of tuples in $\CFL$; viz., context-free languages may only compare distances between variables, whereas the tuples defined by unary $\calV_n$-languages count positions from the beginning of a word. This difference matters only for language classes that are subject to similar restrictions as the context-free languages (e.\,g., the regular languages are not capable of counting, the context-sensitive languages have the ability to convert between these two representations). To account for this special behavior, we will render precisely $\CFL^\N$ in Theorem~\ref{rhm:cfl^n}.
  
  But there is more to be taken into account. Consider, e.\,g., the relation $R = \{ (x,x,x) : x \in \N \}$. $R$ is clearly definable in $\CFL^\N$, yet the set $E(L)$ of the defining language $L$, $L^\N=R$, is not stratified semilinear. Specifically, duplicate variables and permutations of the variables do not increase the complexity of a unary $\calV_n$-language $L$ but affect $L^\N$.
  
  Let $t$ be an order type of $\vec{x}=(x_1,\ldots,x_n)$ and say that a relation $R \subseteq \N^n$ has order type $t$ if, for all $\vec{x} \in R$, $\vec{x}$ has order type $t$.
  For $\vec{x}$ of order type $t$, let $\vec{x}'=(x_1',\ldots,x_m')$, $m \leq n$, denote the variables in $\vec{x}$ with mutually distinct values and let $\pi_t$ denote a permutation such that $x_{\pi_t(i)}' < x_{\pi_t(i+1)}'$, $1 \leq i <m$. 
  We define functions $\mathit{sort}\colon \powerset{\N^n} \to \powerset{\N^m}$ and $\mathit{diff}\colon\powerset{\N^n} \to \powerset{\N_0^n}$ as 
  \begin{align*}
    \mathit{sort}(R) & = \big\{ \pi_t(\vec{x}') : \vec x \in R \text{ has order type $t$} \big\}, \\
    \mathit{diff}(R) & = \Big\{(x_i)_{1\leq i \leq n} : \Big(\sum_{j=1}^i x_j\Big)_{1\leq i \leq n} \in R \Big\}.
  \end{align*}
  The function $\mathit{sort}$ rearranges the components of $R$ in an ascending order and eliminates duplicates, whereas $\mathit{diff}$ transforms a tuple $(x_1,\ldots,x_n)$ with $x_1 < x_2 < \cdots < x_n$ into $(x_1,x_2-x_1,x_3-x_2-x_1, \ldots,x_n - \sum_{i=1}^{n-1} x_i)$, a representation more ``suitable'' to $\CFL$ (cf. $E(L)$ in Theorem~\ref{thm:ginsburg}).
  
  \begin{theorem} \label{rhm:cfl^n}
    Let  $R \subseteq \N^n$. $R\in \CFL^\N$ if and only if there exists a partition $R=R_1 \cup \cdots \cup R_k$ such that each
    $\mathit{diff} \big(\mathit{sort}(R_i)\big)$, $1 \leq i \leq k$, is a stratified semilinear set.
  \end{theorem}

  \begin{proof}
    For the direction from left to right, let $L^\N \in \CFL^\N$, $L^\N \subseteq \N^n$ and let $t_1, \ldots, t_p$ exhaust the possible order types of $\vec{x}=(x_1,\ldots,x_n)$.
    As $\CFL$ is closed under intersection with regular languages, $L$ can be partitioned into context-free languages $L=L_{t_1} \cup \cdots \cup L_{t_p}$ such that $L_{t_i}^\N$ has order type $t_i$, $1 \leq i \leq p$.
    
    Fix any $L_{t_i}$. 
    By design of $\mathit{sort}$, it holds that $x_1 < x_2 < \cdots < x_m$ for all $(x_1,\ldots,x_m) \in \mathit{sort}(L_{t_i}^\N)$, 
    hence $\mathit{diff} \big(\mathit{sort} (L_{t_i}^\N)\big) \in \N_0^m$ is defined.
    Since $\vec x$ has order type $t_i$ for all $\vec x \in L_{t_i}$, $\mathit{sort}(L_{t_i}^\N)=\left(\varphi(L_{t_i})\right)^\N$ for a homomorphism $\varphi$ substituting the characters $(0,V)$, $V \neq \emptyset$, with appropriate $(0,V')$. We thus obtain that  $\mathit{sort}(L_{t_i}^\N) \in \CFL^\N$.
    Say $A^\N=\mathit{sort}(L_{t_i}^\N)$ for the context-free unary $\calV_n$-language $A$, then Theorem~\ref{thm:ginsburg} implies that the set $E(A)$ is a stratified semilinear set. 
    Consider the finite state transducer $T$ in Figure~\ref{fig:transducer}. 
    \begin{figure}[!hbt]
      \centering
      \includegraphics{ext_uniformity.0}
      \caption{Transducer $T$\label{fig:transducer}}
    \end{figure}
    $T$ defines the rational transduction $\psi\colon \Gamma_n^\star \to \{\letterA_1,\ldots,\letterA_m, \letterE\}^\star$, $\psi(\underline{w}) = v_1 \cdots v_{s}$, where $\underline{w}=\underline{w}_1 \cdots \underline{w}_s$, $|\underline{w}|=s$, is the given unary $\calV_n$-structure and, for $1 \leq i \leq s$,  $1 \leq j \leq m$,
    \[
      v_i = 
      \begin{cases}
        \letterA_j, &\text{if } \underline{w}_{i} \cdots \underline{w}_s= (0,\emptyset)^l(0,\{x_j\})\underline{u} \text{ for some } l \in \N_0, \underline{u} \in \Gamma_n^\star, \\
        \letterE, &\text{if } \underline{w}_1 \cdots \underline{w}_i(0,\emptyset)^l=\underline{w} \text{ for some } l \in \N_0.
      \end{cases}
    \]
    We claim that $\mathit{diff}\big(\mathit{sort}(L_{t_i}^\N)\big)=E\big(\psi(A)\big)$. The claim concludes the direction from left to right, since $\psi(A) \in \CFL$ due to the closure of $\CFL$ under rational transductions.
                
    \begin{claim}
      Let $A \in \CFL$ such that $A^\N=\mathit{sort}(L_{t_i}^\N)$, then $\mathit{diff}\big(\mathit{sort}(L_{t_i}^\N)\big)=E\big(\psi(A)\big)$. 
    \end{claim}
    
    To prove the claim, let $\calV_n=\{x_1,\ldots,x_m\}$ and let $A \in \CFL$ be a unary $\calV_n$-language such that $A^\N=\mathit{sort}(L_{t_i}^\N)$. Fix an arbitrary $\vec{c}=(c_1,\ldots,c_m) \in A^\N$ and choose $\underline{w} \in A$ such that $\underline{w}^\N=\vec{c}$. Then $c_1 < c_2 < \cdots < c_m$ and
    \begin{multline*}
      \underline{w}= 
      (0,\emptyset)^{c_1-1}(0,\{x_1\})(0,\emptyset)^{c_2-c_1-1}(0,\{x_2\}) \cdots \\ 
      (0,\emptyset)^{c_{m}-\sum_{i=1}^{m-1}c_i-1}(0,\{x_m\})(0,\emptyset)^{d},
    \end{multline*}
    where $d=|\underline{w}|-c_m$. Hence $\psi(\underline{w}) = \letterA_1^{c_1} \letterA_2^{c_2-c_1} \cdots \letterA_m^{c_{m}-\sum_{i=1}^{m-1}c_i} \letterE^d$ and 
    \[
            E\big(\psi(\{\underline{w}\})\big)=\Big\{\Big(c_1,c_2-c_1,\ldots, c_{m}-\sum_{i=1}^{m-1}c_i\Big)\Big\}.
    \]
    On the other hand, $\mathit{diff}(\{\vec{c}\})=\{(c_1,c_2-c_1,\ldots, c_{m}-\sum_{i=1}^{m-1}c_i)\}$. Thus, for every $\vec{c} \in \mathit{sort}(L_{t_i}^\N)$, $E\big(\psi(\{\vec{c}\})\big)=\mathit{diff}(\{\vec{c}\})$ and $\mathit{diff}\big(\mathit{sort}(L_{t_i}^\N)\big)=E\big(\psi(A)\big)$. 
    This implies the claim and concludes the direction from left to right.
            
    For the direction from right to left, it suffices to show that $R_i \in \CFL^\N$ for each $1 \leq i \leq k$. 
    By assumption, $\mathit{diff} \big( \mathit{sort}(R_i)\big) \subseteq \N_0^m$ is a stratified semilinear set. Thence there exists a bounded language $A \in \CFL$ such that $E(A)=\mathit{diff} \big( \mathit{sort}(R_i)\big)$. Let $A$ w.\,l.\,o.\,g. be bounded by $\letterA_1,\ldots,\letterA_m \in \Sigma$, i.\,e., $A \subseteq \letterA_1^\star\cdots\letterA_m^\star$. Define the rational transduction $\chi\colon\{\letterA_1,\ldots,\letterA_m\}^\star \to \Gamma_n^\star$ as $\chi(w)=\underline{v}_1\cdots \underline{v}_s$, where $w=w_1\cdots w_s$, $|\underline{w}|=s$, and, for $1 \leq i \leq s$, $1 \leq j \leq m$,
    \[
      v_i = 
      \begin{cases}
              (0,\{x_j\}), &\text{if } i=s \text{ and } w_i = \letterA_j \text{ or } i < s \text{ and } w_i = \letterA_j \neq  w_{i+1}, \\
              (0,\emptyset) &\text{if } i<s \text{ and } w_i = w_{i+1}. \\
      \end{cases}
    \]
    Note that $(\psi\circ\chi)(w)=w$ for all $w \in A$. An argument analogous to the above claim thus yields $E(A)=E\big(\psi(\chi(A))\big) = \mathit{diff}\big((\chi(A))^\N\big)$ and $\mathit{sort}(R_i)=(\chi(A))^\N$.
    In particular, $(\chi(A))^\N \in \CFL^\N$.
    
    Moreover, $c_1 < c_2 < \cdots < c_m$ for all $\vec{c}=(c_1,\ldots, c_m) \in (\chi(A))^\N$, thus there exists a function 
    $\pi\colon\{1,\ldots,n\}\to \{1,\ldots, m\}$, $n \geq m$, such that 
    \[
      \big\{ (x_{\pi(1)},\ldots,x_{\pi(n)}) : (x_1,\ldots, x_m) \in (\chi(A))^\N \big\} = R_i.
    \]
    Let the homomorphism $\phi\colon\Gamma_n^\star \to \Gamma_n^\star$ mimic the above transformation by replacing $(0,\{x_i\})$ with $(0,V_i)$, where $V_i=\{x_j: 1 \leq j \leq n, \pi(j)=i\}$, $1 \leq i \leq m$. Then $R_i = \big((\phi \circ\chi)(A)\big)^\N \in \CFL^\N$.
  \end{proof}

\section{The Uniformity Duality and Context-Free Languages} \label{sect:cfl_property}
    
  Due to the previous section, we may express the Uniformity Duality Property for context-free languages using Corollary~\ref{cor:bccfl subset fo+} in the following more intuitive way: 
  let $\calQ=\{\exists\}$ and $\calL$ be such that 
  $\FO[\calL^\N] = \FO[\mathord{<}\nolinebreak,+]$
  (e.\,g., $\DCFL \subseteq \calL \subseteq \bc{\CFL}$), then the Uniformity Duality Property for $(\{\exists\},\calL)$ is equivalent to
  \begin{equation} \label{eq:prop_w_cfl}
    \FO[\arb] \cap \calL = \FO[\mathord{<}\nolinebreak, +] \cap \calL.
  \end{equation}
  We will hence examine whether \eqref{eq:prop_w_cfl} holds, and see that this is not the case.
  
  For a binary word $u=u_{n-1}u_{n-2}\cdots u_0 \in \{0,1\}^\star$, we write $\integer{u}$ for the integer $u_{n-1} 2^{n-1} + \cdots + 2 u_1 + u_0$. Recall the Immerman language $L_I \subseteq \{0, 1, \letterA\}^\star $, that is, the language consisting of all words of the form
  $$
    x_1 \, \letterA \, x_2 \, \letterA \, \cdots \, \letterA \, x_{2^n},
  $$
  where $x_i \in \{0,1\}^n$, $\integer{x}_i + 1 = \integer{x}_{i+1}$, $1 \leq i < 2^n$, and $x_1=0^n$, $x_{2^n}= 1^n$. For example, 
  $00 \letterA 01 \letterA 10 \letterA 11  \in L_I$ and $000 \letterA 001 \letterA 010 \letterA 011 \letterA 100 \letterA 101 \letterA 110 \letterA 111 \in L_I$. We prove that despite its definition involving arithmetic, $L_I$ is simply the complement of a context-free language.
    
  \begin{lemma} \label{lem:immerman}
    The complement $\complement{L_I}$ of the Immerman language is context-free.
  \end{lemma}
  
  \begin{proof}
    Let $\Sigma=\{0, 1, \letterA\}$. Throughout this proof, $u$ and $v$ stand for binary words.
    \begin{claim}
      Let $u=u_{n-1}u_{n-2}\cdots u_0, v=v_{n-1}v_{n-2} \cdots v_0$ and $u_0\neq v_0$. Then
      \begin{equation}\label{eq:succ}
        \integer{u} +1 =  \integer{v} \pmod{2^n}
      \end{equation}
      iff none of the words $u_1u_0v_1v_0, u_2u_1v_2v_1, \ldots,$ $u_{n-1}u_{n-2}v_{n-1}v_{n-2}$ belongs to
      \begin{equation}\label{eq:succcondition}
        \{ 0010, 0011, 0100, 0111, 1001, 1000, 1110, 1101 \}. 
      \end{equation}
    \end{claim}
    The claim implies that the following language is context-free:
    \[
      A = \{ xu \letterA vy : x\in\Sigma^\star , y\in\Sigma^\star, |u|=|v|, \integer{u} + 1 \neq \integer{v} \pmod{2^{|u|}} \}.
    \]
    Note that $\{ xu \letterA vy : x\in (\Sigma^\star \letterA)^\star,\ y \in (\letterA \Sigma^\star)^\star,\ |u|=|v|,\ \integer{u} + 1 \neq \integer{v}  \pmod{2^{|u|}} \} \subset A \subset \complement{L_I}$. Hence $A$ does not catch all the words in $\complement{L_I}$, but it catches all the words of the correct ``form'' which violate the successor condition (modulo $2^{|u|}$). For example, $A$ catches $00 \letterA 01 \letterA 11 \letterA 11 \letterA$, but it does not catch $\letterA$ nor $0 \letterA 1 \letterA 0 \letterA 1$ nor   $01 \letterA 10 \letterA 11 \letterA 00$ nor $00 \letterA 01 \letterA 10 \letterA 11 \letterA 00 \letterA 01$ nor $00 \letterA 01001 \letterA 10 \letterA 11$.
    We complete the proof by expressing $\complement{L_I}$ as follows:
    \begin{equation} \label{eq:L_Icomplement}
      \renewcommand{\arraystretch}{1.15}
      \begin{array}{rl} 
        \complement{L_I} = 
          & 
          A                                                              \cup 
          \letterA^\star                                                 \cup 
          \Sigma^\star \letterA 0^\star \letterA \Sigma^\star            \cup
          \Sigma^\star \letterA 1^\star \letterA \Sigma^\star \,         \cup  \\ &
          \{0,1\}^\star 1 \Sigma^\star \cup \Sigma^\star 0 \{0,1\}^\star \cup
          \bigcup_{|u|\neq |v|} (\Sigma^\star \letterA)^\star u \letterA v (\letterA \Sigma^\star )^\star.
      \end{array}
    \end{equation}
    
    Now we prove the claim. The direction from left to right is easy: if any word $u_iu_{i-1}v_iv_{i-1}$ belongs to the forbidden set \eqref{eq:succcondition} then \eqref{eq:succ} trivially fails.        For the converse, we must prove that if no forbidden word occurs then   \eqref{eq:succ} holds.  We prove this by induction on   the common length $n$ of the words $u$ and $v$. When $n=1$, $u_0$ and $v_0\neq u_0$ are successors modulo $2$. So suppose that $n\geq 1$ and that none of the forbidden words occurs in the pair $(u_nu,v_nv)$, where $u_n,v_n\in\{0,1\}$, $u=u_{n-1}\cdots u_1u_0$ and $v=v_{n-1}\cdots v_1v_0$ and $u_0\neq v_0$. Then by induction, $\integer{u} + 1 = \integer{v} \pmod{2^n}$. There are four cases to be treated:
    
    \renewcommand{\descriptionlabel}[1]{\hspace\labelsep\normalfont{#1}}
    \begin{description}
      \item[Case 1: $u_{n-1}=v_{n-1}=0$.] Then no overflow into $u_n$ occurs when $1$ is added to $\integer{u}$ to obtain $\integer{v}$.  Since the forbidden words leave only $u_n=v_n$ as possibilities, $\integer{u_n u} + 1 = \integer{v_n v}$. 
      \item[Case 2: $u_{n-1}=v_{n-1}=1$.] Analogous.
      \item[Case 3: $0=u_{n-1}\neq v_{n-1}=1$.] Analogous.
      \item[Case 4: $1=u_{n-1}\neq v_{n-1}=0$.] This is an interesting case. The fact that $\integer{u} + 1 = \integer{v} \pmod{2^n}$ implies
        that $u=1^n$ and $v=0^n$. Now the forbidden words imply $u_n\neq v_n$. This means that either $u_nu=01^n$ and $v_nv=10^n$, or 
        $u_nu=1^{n+1}$ and $v_nv=0^{n+1}$. In the former case, $\integer{u_nu} + 1 = \integer{v_nv}$, and in the latter case, $\integer{u_nu} + 1 = \integer{v_nv} \pmod{2^{n+1}}$.
    \end{description}
  \end{proof}
  
  For a language $L \subseteq \Sigma^\star$, let $\neutral{L}$ denote $L$ supplemented with a neutral letter $\letterE \notin \Sigma$, i.e., $\neutral{L}$ consists of all words in $L$ with possibly arbitrary repeated insertions of the neutral letter.
  
  \begin{theorem} \label{thm:dual_prop_fails_for_bccfl}
    $\FO[\arb]\cap \bc{\CFL} \supsetneq \FO[\mathord{<}\nolinebreak,+] \cap \bc{\CFL}.$
  \end{theorem}
  
  \begin{proof}
    From the Crane Beach Conjecture by Barrington et~al. 
    \cite[Lemma~5.4]{baimlascth05}, we know that $\neutral{L_I} \in \FO[\arb] \setminus \FO[\mathord{<}\nolinebreak,+]$ . So we are done if we can show $\neutral{L_I} \in \bc{\CFL}$. We proved in Lemma~\ref{lem:immerman} that $\complement{L_I}=\{0, 1, \letterA \}^\star \setminus L_I$ is context-free. Therefore $\neutral{\complement{L_I}}$ is context-free. Now, for any $w_\letterE \in  \{0, 1, \letterA, \letterE\}^\star $, let $w \in \{0, 1, \letterA\}^\star $ be the word obtained by deleting all occurrences of the neutral letter $\letterE$ from $w_\letterE$. Then for any $w_\letterE$,
    $$
      w_\letterE\in \neutral{L_I} \iff w \in L_I
            \iff w \notin \complement{L_I}
            \iff w_\letterE \notin \neutral{\complement{L_I}}.
    $$
    In other words, $\neutral{L_I} = \{0, 1, \letterA, \letterE\}^\star \setminus \neutral{\complement{L_I}} = \complement{\neutral{\complement{L_I}}}$ and thus $\neutral{L_I} \in \bc{\CFL}$.
  \end{proof}

  Theorem~\ref{thm:dual_prop_fails_for_bccfl} implies that the Uniformity Duality Property fails for $\calQ=\{\exists\}$ and $\calL=\bc{\CFL}$, 
  since $\FO[\mathord{<}, \bc{\CFL}^\N] = \FO[\mathord{<}, \mathord{+}]$.
  Yet, it even provides a witness for the failure of the duality property in the case of $\calL = \CFL$, as the context-free language $\neutral{\complement{L_I}}$ lies in $\FO[\arb] \setminus \FO[\mathord{<},+]$. We will state this result as a corollary further below. 
  For now, consider the modified Immerman language $R_I$ defined as $L_I$ except that the successive binary words are reversed in alternance, i.\,e.,
  $$
    R_I = \{ \ldots, 000 \letterA (001)^R \letterA 010 \letterA (011)^R \letterA 100 \letterA (101)^R \letterA 110 \letterA (111)^R,\ldots\}.
  $$
  $R_I$ is the intersection of two deterministic context-free languages. Even more, the argument in Lemma~\ref{lem:immerman} can actually be extended to prove that the complement of $R_I$ is a linear $\CFL$. Hence,
  
  \begin{theorem} \label{thm:dual_prop_fails_for_bcdcfl}
    \begin{enumerate}
    \item $\FO[\arb]\cap \bc{\DCFL} \supsetneq \FO[\mathord{<}\nolinebreak,+] \cap \bc{\DCFL}$.
    \item $\FO[\arb]\cap \bc{\mathrm{Lin}\CFL} \supsetneq \FO[\mathord{<}\nolinebreak,+] \cap \bc{\mathrm{Lin}\CFL}$.
    \end{enumerate}
  \end{theorem}

  \begin{proof}
    For the first claim, observe that $\neutral{R_I}\notin
    \FO[\mathord{<},\mathord{+}]$, because $\FO[\mathord{<},\mathord{+}]$ has the Crane
    Beach Property and $R_I\notin\REG$. On the other hand,
    $\neutral{R_I}\in\FO[\arb]$; and since $R_I\in  \bc{\DCFL}$,
    $\neutral{R_I}\in \bc{\DCFL}$.
    
    For the second claim, $\complement{R_I}$ can be expressed analogously to $\complement{L_I}$ by substituting $A$ with an appropriate set $A'$ in \eqref{eq:L_Icomplement}. Further, for $u\letterA v \in \Sigma^\star$ with binary words $u$ and $v$, the conditions $\integer{u^R} + 1 \neq \integer{v} \pmod{2^{|u|}}$ and $\integer{u} + 1 \neq \integer{v^R} \pmod{2^{|u|}}$ can be checked by a linear $\CFL$. Since the linear context-free languages are closed under finite union, the claim follows.
  \end{proof}

  The role of neutral letters in the above theorems suggests taking a
  closer look at $\neutral{\CFL}$. As the Uniformity Duality Property
  for $(\{\exists\},\neutral{\CFL})$ would have it,
  all neutral-letter context-free languages in $\AC{0}$ would be
  regular and aperiodic. This is, however, not the case as witnessed by
  $\neutral{\complement{L_I}}$. Hence,
  
  \begin{corollary}
    In the case of $\calQ=\{\exists\}$, the Uniformity Duality
    Property fails in all of the following cases.
    \begin{enumerate}
      \item $\calL=\CFL$,
      \item $\calL=\bc{\CFL}$,
      \item $\calL=\bc{\DCFL}$,
      \item $\calL=\bc{\mathrm{Lin}\CFL}$,
      \item $\calL=\neutral{\CFL}$.
    \end{enumerate}
  \end{corollary}
  
  \begin{remark}
    The class $\VPL$ of visibly pushdown languages \cite{alma04} has gained
    prominence recently because it shares with $\REG$ many useful
    properties. 
    But despite having access to a stack, the $\VPL$-numerical
    predicates coincide with $\REG^\N$, for each word may only contain 
    constantly many characters different from $(0,\emptyset)$. 
    It follows that the Uniformity Duality Property fails for $\VPL$ and 
    first-order quantifiers: consider, e.\,g., 
    $L=\{\letterA^n\letterB^n : n > 0\} \in \FO[\arb] \cap (\VPL\setminus\REG)$
    then $L \in \FO[\arb] \cap \VPL$ but $L \notin \FO[\mathord{<},\VPL^\N] 
    \cap \VPL$.
  \end{remark}

\section{The Duality in Higher Classes} \label{sect:csl_property}

  We have seen that the context-free languages do not exhibit our
  conjectured Uniformity Duality. In this section we will show that
  the Uniformity Duality Property holds if the extensional uniformity
  condition imposed by intersecting with $\calL$ is quite loose,
  in other words, if the language class $\calL$ is powerful.
  
  Recall the notion of non-uniformity introduced by Karp and Lipton \cite{kali82}.
  
  \begin{definition}
    For a complexity class $\calL$, denote by $\calL\withPolyAdvice$ the class $\calL$ with polynomial advice. That is, $\calL\withPolyAdvice$ is the class of all languages $L$ such that, for each $L$, there is a function $f\colon \N \to \{0,1\}^\star$ with 
    \begin{enumerate}
      \item $|f(x)| \leq p(|x|)$, for all $x$, and
      \item $L^f = \{ \langle x, f(|x|) \rangle : x \in L \} \in \calL$,
    \end{enumerate}
    where $p$ is a polynomial depending on $\calL$. Without loss of generality, we will assume $|f(x)| = |x|^k$ for some $k \in \N$.
  \end{definition}
  Note that, using the above notation, $\DLOGTIME$-uniform $\AC{0}\withPolyAdvice = \AC{0}$.
  As we further need to make the advice strings accessible in a logic, we define the following predicates.
  
  Following~\cite{baimst90}, we say that a Lindstr\"om quantifier $Q_L$ is \emph{groupoidal}, if $L \in \CFL$.
  
  \begin{definition} \label{def:advice_predicate}
    Let $\calQ$ be any set of groupoidal quantifiers. Further, let $L \in \DLOGTIME$-uniform $\AC{0}[\calQ]\withPolyAdvice$ and let $f$ be the function for which $L^f \in \DLOGTIME$-uniform $\AC{0}[\calQ]$. Let $r=2kl+1$, where $k$ and $l$ are chosen such that the circuit family recognizing $L$ in $\DLOGTIME$-uniform $\AC{0}$ has size $n^l$ and $|f(x)|=|x|^k$. We define $\textsc{Advice}_{L,\calQ}^f \in \FO\mathord{+}\calQ[\arb]$ to be the ternary relation
    $$
      \textsc{Advice}_{L,\calQ}^f = \{ (i,n,n^r) : \text{bit $i$ of $f(n)$ equals $1$} \}, 
    $$
    and denote the set of all relations $\textsc{Advice}_{L,\calQ}^f$, for $L \in \calL$, by $\textsc{Advice}_{\calL,\calQ}$.
  \end{definition}
  
  The intention of $\textsc{Advice}_{L,\calQ}^f$ is to encode the advice string as a numerical relation. A point in this definition that will become clear later is the third argument of the $\textsc{Advice}_{L,\calQ}^f$-predicate; it will pad words in the corresponding unary $\calV_n$-language to the length of the advice string. This padding will be required for Theorem~\ref{thm:ufo_holds_for_dspace(n)}.
  
  \begin{theorem} \label{thm:udp-holds-if-advice-is-def}
    Let $\calL$ be a language class and $\calQ$ be a set of groupoidal quantifiers. Then the Uniformity Duality Property for $(\{\exists\}\cup\calQ, \calL)$ holds if $\bit \in \calL^\N$ and $\textsc{Advice}_{\calL,\calQ}\in \calL^\N$. 
  \end{theorem}
  
  \begin{proof}
    Let $\calL$ be a language class that satisfies the requirements of the claim. We have to show that $\FO\mathord{+}\calQ[\arb]\cap \calL = \FO\mathord{+}\calQ[\mathord{<}\nolinebreak,\calL^\N] \cap \calL$.
    
    The inclusion from right to left is trivial. For the other direction, let $L \in \FO\mathord{+}\calQ[\arb] \cap \calL$. Without loss of generality, we assume $L \subseteq \{0,1\}^\star$.
    In \cite{baimst90}, Barrington et~al. state that $\FO\mathord{+}\calQ[\arb] = \FO\mathord{+}\calQ[\bit]\withPolyAdvice = \DLOGTIME$-uniform $\AC{0}[\calQ]\withPolyAdvice$ for arbitrary sets $\calQ$ of monoidal quantifiers, but what is needed in fact is only the existence of a neutral element, which is given in our case.
    There hence exists a polynomial $p(n)=n^k$, $k \in \N$, a function $f$ with $|f(x)| = p(|x|)$, and a $\DLOGTIME$-uniform circuit family $\{C_m\}_{m>0}$ that recognizes $L^f$. 
    From $\{C_m\}_{m>0}$, we construct a formula $\varphi \in \FO\mathord{+}\calQ[\mathord{<}\nolinebreak,\calL^\N]$, essentially replacing the advice input gates with the relation $\textsc{Advice}_{L,\calQ}^f$. 
    
    Let $x_1, \ldots, x_n, y_1, \ldots, y_{n^k}$ denote the input gates of circuit $C_{m}$, $m>0$, where $x_1\cdots x_n=x$ and $y_1\cdots y_{n^k}=f(|x|)$. 
    First canonically transform $\{C_m\}_{m>0}$ into a formula $\varphi'$ over an extended vocabulary $\sigma$ that satisfies $L(\varphi') = L$ (cf.  \cite[Theorem 9.1]{baimst90}, \cite[Theorem IX.2.1]{str94} or \cite[Theorem 4.73]{vol99} for the construction of $\varphi'$). 
    Let $l \in \N$ be such that $n^l$ is a size bound on $\{C_m\}_{m>0}$. Then
    the transformation encodes gates as $l$-tuples of variables over $\{1, \ldots, n\}$ and ensures the correct structure using additional predicate symbols $\textsc{Pred}$, $\textsc{Input0}$, $\textsc{Input1}$, $\textsc{Output}$, $\textsc{And}$, $\textsc{Or}$, $\textsc{Not}$ and predicates $\textsc{Quant}_Q$ for the subset of oracle gates from $\calQ$ used in $C_m$.
    For example, the relation $\textsc{Pred}$ holds on a tuple $(z_1,\ldots, z_{2l})$ iff the gate encoded by $(z_1,\ldots, z_l)$ is a predecessor of $(z_{l+1},\ldots, z_{2l})$; and each of remaining predicates holds on a tuple $(z_1,\ldots, z_l)$ iff $(z_1,\ldots, z_l)$ encodes a gate of the corresponding type. Note that each relation in $\sigma$ is definable in $\FO[\mathord{<},\calL^\N]$, as $\DLOGTIME \subseteq \FO[\bit] \subseteq \FO[\mathord{<},\calL^\N]$.
    
    Next, replace the relations $\textsc{Input0}$ and $\textsc{Input1}$ corresponding to the input gates $y_i$, $1 \leq i \leq n^k$, with the respective advice predicates from Definition~\ref{def:advice_predicate},
    $$
      \neg \textsc{Advice}_{L,\calQ}^f(i,n,n^r)
    \quad\text{and}\quad
      \textsc{Advice}_{L,\calQ}^f(i,n,n^r);
    $$ 
    both of which are definable in $\calL^\N$ by assumption. In the resulting formula, replace the remaining predicates besides $<$ and $\bit$ by their defining $\FO\mathord{+}\calQ[\mathord{<},\calL^\N]$-formulae and eventually obtain a formula defining $L \in \FO\mathord{+}\calQ[\mathord{<},\calL^\N]$.
  \end{proof}

  We can now give a lower bound beyond which the Uniformity Duality Property holds. Let $\NTIME{n}^\calL$ denote the class of languages decidable in linear time by nondeterministic Turing machines with oracles from $\calL$.

  \begin{theorem} \label{thm:ufo_holds_for_dspace(n)}
    Let $\calQ$ be any set of groupoidal quantifiers and suppose $\calL = \NTIME{n}^\calL$. Then the Uniformity Duality Property for $(\{\exists\} \cup \calQ,\calL$) holds.
  \end{theorem}

  \begin{proof}
    Choose any $L \in \AC{0}[\calQ] \cap \calL$ and let $f$ be an advice function for which 
    $L^f \in \DLOGTIME$-uniform $\AC{0}[\calQ]$. We have to show that the relation
    $\textsc{Advice}_{L,\calQ}^f$ is definable in $\calL^\N$. Let $k \in \N$ 
    such that $|f(x)|=|x|^k$ and denote by $\{C_m\}_{m>0}$ the
    $\DLOGTIME$-uniform circuit family of size $\leq n^l$ that recognizes $L^f$ using oracle gates from $\calQ$. 
    Let $N$ be a nondeterministic linear-time Turing machine deciding $L$ using some oracle $L' \in \calL$.  
    
    We will define a nondeterministic Turing machine $M$ that decides
    $x \in L_\text{Adv}$, where $L_\text{Adv}$ is such that
    $L_\text{Adv}^\N = \textsc{ADVICE}_{L,\calQ}^f$. Given input $x$, $M$
    proceeds as follows: 
    \begin{lstlisting}[gobble=6, tabsize=1]
      if $x$ is not of the form $x=[x_1=i,x_2=n,x_3=n^r]$ $\text{for}$ some $n>0 \text{ and } 1 \leq i \leq n$ 
      then reject;
      for all strings $a$ of length $n^k$ in lexicographic ordering do
        $t$ := true;
        for all inputs $y$ of length $n$ do
          if the output of $C_{n+n^k}$ on $y$ with advice $a$ contradicts the result of $N$ on $y$ (*\label{lst:circuit}*)
          then $t$ := false; 
        if $t$ = true and bit $i$ of $a$ is $1$ then accept;
        else reject;
    \end{lstlisting}
    That is, $M$ guesses the advice string $a$ using a na\"ive trial-and-error approach; once the correct advice string $a$ has been found, it accepts $x$ iff the $i$th bit in the advice string $a$ is on. Thus $M$ decides $L_\text{Adv}$.
    
    As for the time required by $M$, note that in line~\ref{lst:circuit} the circuit $C_{n+n^k}$ can be evaluated in time $\bigO{(n+n^k)^{2l}}=\bigO{n^r} = \bigO{|x|}$. 
    Furthermore, $\calL = \NTIME{n}^\calL$ implies that $\calL$ is closed under complement. Thus the above algorithm solves the problem in $\NTIME{|x|}^{\co\NTIME{|x|}^\calL} = \NTIME{|x|}^{\calL}= \calL$ and $\textsc{ADVICE}_{L,\calQ}^f \in \calL^\N$.
    The claim now follows from Theorem~\ref{thm:udp-holds-if-advice-is-def}, because $\calL = \NTIME{n}^\calL$ moreover implies $\bit \in \calL^\N$.
  \end{proof}
  
  \begin{corollary} \label{cor:ufo_holds_for_csl}
    Let $\calQ$ be any set of groupoidal quantifiers.
    The Uniformity Duality Property holds for $(\{\exists\} \cup \calQ,\calL)$ if $\calL$ equals the deterministic context-sensitive languages $\DSPACE{n}$, the context-sensitive languages $\CSL$, the rudimentary languages (i.\,e., the linear time hierarchy \cite{wra78}), $\PH$, $\PSPACE$, or the recursively enumerable languages.
  \end{corollary}
  
 \begin{proof}
   All of the above classes satisfy $\calL=\NTIME{n}^\calL$.
 \end{proof}

\section{Conclusion} \label{sect:conclusion}
  
  For a set $\calQ$ of quantifiers and a class $\calL$ of languages,
  we have suggested that $\calQ[\arb]\cap \calL$ defines
  an (extensionally) uniform complexity class.
  After defining the notion of $\calL$-numerical predicates,
  we have proposed comparing $\calQ[\arb]\cap \calL$ with its subclass
  $\calQ[\mathord{<},\calL^\N]\cap \calL$, a class equivalently defined as
  the (intensionally) uniform circuit class $\FO[\mathord{<},\calL^\N]$-uniform
  $\AC{0}[\calQ] \cap \calL$.

  We have noted that the duality property, defined to hold when both
  classes above are equal, encompasses Straubing's conjecture
  \eqref{conj:reg} as well as some positive and some negative instances
  of the Crane Beach Conjecture.
  
  We have then investigated the duality property in specific cases with
  $\calQ = \{\exists\}$. We have seen that the property fails for
  several classes $\calL$ involving the context-free languages.
  Exhibiting these failures has required new
  insights, such as characterizations of the context-free numerical
  predicates and a proof that the complement of the Immerman
  language is context-free, but these failures have prevented
  successfully tackling complexity classes such as $\AC{0} \cap \CFL$.
  Restricting the class of allowed relations on the left hand side of the uniformity duality property from $\arb$ to a subclass might lead to further insight and provide positive examples of this modified duality property (and address, e.g., the class of context-free languages in different uniform versions of $\AC{0}$). Methods from embedded finite model theory should find applications here. 
  
  More generally, the duality property widens our perspective on the
  relationship between uniform circuits and descriptive complexity
  beyond the level of $\NC{1}$.
  We have noted for example that the property holds for any set of groupoidal quantifiers $\calQ \supseteq \{\exists\}$ and complexity
  classes $\calL$ that are closed under nondeterministic linear-time Turing reductions.
  
  A point often made is that a satisfactory uniformity definition should
  apply comparable resource bounds to a circuit family and to its
  constructor.
  For instance, although $\P$-uniform $\NC{1}$ has merit \cite{all89}, 
  the classes $\AC{0}$-uniform $\NC{1}$ and $\NC{1}$-uniform
  $\NC{1}$ \cite{baimst90} seem more fundamental, provided that
  one can make sense of the apparent circularity.
  As a by-product of our work, we might suggest $\FO[\mathord{<},\calL^\N] \cap
  \calL$ as 
  the minimal ``uniform subclass of $\calL$'' and thus
  as a meaningful (albeit restrictive) 
  definition of \emph{$\calL$-uniform $\calL$}.
  Our choice of $\FO[\mathord{<}]$ as the ``bottom class of interest'' is implicit
  in this definition and results in the containment of
  \emph{$\calL$-uniform $\calL$} in (non-uniform) $\AC{0}$ for
  any $\calL$.
  Progressively less uniform subclasses of $\calL$ would be the classes
  $\calQ[\mathord{<},\calL^\N] \cap \calL$  for $\calQ \supseteq \{\exists\}$.

  Restating hard questions such as conjecture \eqref{conj:reg} in terms of a unifying property does not make these questions go away. But the duality property raises further questions. As an example, can the duality property for various $(\calQ,\calL)$ be shown to hold or to fail when $\calQ$ includes the majority quantifier? This could help develop incisive results concerning the class $\TC{0}$.
  To be more precise, let us consider $\calQ=\{\exists,\MAJ\}$. The majority quantifier is a particular groupoidal (or, \emph{context-free}) quantifier \cite{lamcscvo01}, hence it seems natural to consider the Uniformity Duality Property for $(\{\exists,\MAJ\},\CFL)$:
  \begin{equation} \label{conj:TC0}
    \FO\mathord{+}\MAJ[\arb]\cap{\CFL}=\FO\mathord{+}\MAJ[\mathord{<},+]\cap{\CFL}.
  \end{equation}
  It is not 
  hard to see that the Immerman language in fact is in $\FO\mathord{+}\MAJ[\mathord{<},+]$, hence our Theorem~\ref{thm:dual_prop_fails_for_bccfl} that refutes \eqref{eq:prop_w_cfl}, the Uniformity Duality Property for $(\FO,\bc{\CFL})$, does not speak to whether \eqref{conj:TC0} holds. (Another prominent example that refutes \eqref{eq:prop_w_cfl} is the ``Wotschke language'' $\text{W}=\{(\letterA^n\letterB)^n : n\geq0 \}$, again a co-context-free language \cite{wotschke73}. Similar to the case of the Immerman language we observe that $\text{W}\in\FO\mathord{+}\MAJ[\mathord{<},+]$, hence $\text{W}$ does not refute \eqref{conj:TC0} either.)

  Observe that $\FO\mathord{+}\MAJ[\arb]=\TC{0}$ \cite{baimst90} and that, on the
  other hand, \linebreak $\FO\mathord{+}\MAJ[\mathord{<},+] =\MAJ[\mathord{<}]=\FO[\mathord{+}]\text{-uniform }
  \text{linear fan-in }\TC{0}$ \cite{la04,bela06}. Let us call this 
  latter class $\sTC{0}$ (for \emph{small} $\TC{0}$ or
  \emph{strict} $\TC{0}$). It is known that
  $\sTC{0}\subsetneq\TC{0}$ \cite{lamcscvo01}. Hence we conclude
  that if \eqref{conj:TC0} holds, 
  then in fact $\TC{0}\cap{\CFL}=\sTC{0}\cap{\CFL}$. 
  Thus, if we can show that some language in the Boolean 
  closure of the context-free languages is  not in $\sTC{0}$, 
  we have a new $\TC{0}$ lower bound. Thus, to separate $\TC{0}$ 
  from a superclass it suffices to separate $\sTC{0}$ from a 
  superclass, a possibly less demanding goal. 
  This may be another reason to look for
  appropriate uniform classes $\calL$ such that
  $$\FO\mathord{+}\MAJ[\arb]\cap\calL=\FO\mathord{+}\MAJ[\mathord{<},+]\cap\calL.$$
  
\section*{Acknowledgements}
  We would like to thank Klaus-J\"{o}rn Lange (personal
  communication) for suggesting Lemma~\ref{lem:immerman}. We
  also acknowledge helpful discussions on various topics of this
  paper with Christoph Behle, Andreas Krebs, Klaus-J\"orn Lange
  and Thomas Schwentick.
  We also acknowledge helpful comments from the anonymous referees.

\bibliographystyle{plain}
\bibliography{thi-hannover}

\end{document}